\documentclass[final,pdftex,reqno]{ectaart}

\RequirePackage[OT1]{fontenc}
\RequirePackage{graphicx}
\RequirePackage{amsthm}
\RequirePackage[cmex10]{amsmath}
\RequirePackage{natbib}
\RequirePackage[colorlinks,citecolor=blue,urlcolor=blue]{hyperref}

\startlocaldefs
\numberwithin{equation}{section}
\theoremstyle{plain}

\theoremstyle{plain}

\newtheorem{theorem}{Theorem}

\newtheorem{corollary}{Corollary}
\newtheorem{lemma}{Lemma}

\newtheorem{pro}{Proposition}
\newtheorem{asu}{Assumption}

\newcounter{subassumption}[asu]

\makeatletter
\renewcommand{\p@subassumption}{\theasu}
\makeatother

\newcounter{subproposition}[pro]

\makeatletter
\renewcommand{\p@subproposition}{\thepro}
\makeatother

\theoremstyle{definition}
\newtheorem{definition}{Definition}

\newtheorem{remark}{Remark}
\newtheorem{model}{Model}
\newtheorem*{psm}{PSM technique}

\endlocaldefs
\usepackage{amssymb}
\usepackage{amsfonts} 
\usepackage{tikz}
\usetikzlibrary{shapes,shadows,arrows}
\usepackage{enumerate}
\usepackage{enumitem}
\usepackage[open]{bookmark}
\usepackage{pgfplots}
\pgfplotsset{compat=newest}
\usepgfplotslibrary{statistics}
\usepgfplotslibrary{fillbetween}
\usepackage[thinc]{esdiff}
\usepackage{xcolor}
\usepackage{relsize}
\usepackage{exscale}
\usepackage{nicefrac}
\usepackage{ragged2e}
\usepackage[utf8]{inputenc}
\usepackage{booktabs}
\usepackage[flushleft]{threeparttable}
\begin{document}
\pdfcatalog{/PageLayout /SinglePage}
\pgfplotsset{
standard/.style={
axis line style=thick,
trig format=rad,
enlargelimits,
axis x line=middle,
axis y line=middle,
enlarge x limits=0.15,
enlarge y limits=0.15,
every axis x label/.style={at={(current axis.right of origin)} ,anchor=north west},
every axis y label/.style={at={(current axis.above origin)} ,anchor=south east}
}
}
\begin{frontmatter}
\title{An Econometric Analysis of the Impact of Telecare on the Length of Stay in Hospital\protect\thanksref{T1}}
\runtitle{Demand for Telecare and Hospital Length of Stay}
\thankstext{T1}{This paper builds on one of the chapters of the author's Ph.D. that investigated the effectiveness of telecare using economic analyses.  A previous version of the paper was presented in $2017$ at the University of Aberdeen as part of the 24-month assessment for the Ph.D.  The empirical study in this paper used administrative data for patients in Scotland.  All the patient names were replaced with unique reference numbers to ensure that the data was anonymized.}

\begin{aug}
\author{\fnms{Kevin} \snm{Momanyi}\thanksref{a,t1}\ead[label=e1]{kevzy44@gmail.com}}

\thankstext{t1}{Corresponding author, email address:\texttt{ kevzy44@gmail.com.}}
\runauthor{K. Momanyi, 2025}

\address[a]{
Health Economics Research Unit, University of Aberdeen
, Polwarth Building, Foresterhill Campus,
Aberdeen AB25 2ZD.
}

\end{aug}

\begin{abstract}
In this paper, we develop a theoretical model that links the demand for telecare to the length of stay in hospital and formulate three models that can be used to derive the treatment effect by making various assumptions about the probability distribution of the outcome measure.  We then fit the models to data and estimate them using a strategy that controls for the effects of confounding variables and unobservable factors, and compare the treatment effects with that of the Propensity Score Matching (PSM) technique which adopts a quasi-experimental study design.  To ensure comparability, the covariates are kept identical in all cases.  An important finding that emerges from our analysis is that the treatment effects derived from our econometric models of interest are better than that obtained from an experimental study design as the latter does not account for all the relevant unobservable factors.  In particular, the results show that estimating the treatment effect of telecare in the way that an experimental study design entails fails to account for the systematic variations in individuals' health production functions within each experimental arm.
\end{abstract}


\end{frontmatter}

\section{Introduction}\label{Section 1}

The use of devices to monitor individuals' health and safety at home--commonly referred to as telecare--may complement or substitute for social care and unpaid care.  Telecare covers a wide range of devices from the basic community alarm, which allows the users to call for help by simply pressing a button, to more sophisticated devices that allow for remote exchange of clinical data between the users and their care providers and virtual consulting using audio and video technology (\cite{r13,r14}).  The use of telecare is beneficial to both the users and their carers.  For instance, telecare reduces the need for residential care mainly through delayed admissions and also offers increased choice and independence for the users.  Telecare can also reduce pressures on carers by freeing up some of their time thus giving them more personal freedom (\cite{r17,r23}).

Although the health economics literature abounds with several outcome measures ranging from functional performance measures such as lower body strength and level of assistance in performing activities of daily living (see, for example, \cite{r34,r46,r35}) to measures of health and well-being such as Quality Adjusted Life Year (QALY) and social care related quality of life (see, for example, \cite{r27,r54,r31}), we consider the length of stay in hospital as our outcome measure of interest.  This measure is ideal as it is an indicator of both individuals' general health and the strength of the health system (see, for example, \cite{r5,r53,r22}).  The length of stay in hospital is also of great interest to health policy makers who aim at designing policies that would reduce the prevalence of health care associated illnesses and general inpatient costs as well as supporting self-management of chronic illnesses (see, for example, \cite{r19,r8}).

In this paper, we present a novel framework for estimating the effect of the demand for telecare on the length of stay in hospital.  We argue that the use of telecare could substitute for some health care services that would have otherwise been provided in hospital, thereby resulting in a shorter length of stay in hospital.  It could be the case, therefore, that telecare users have a comparatively short length of stay in hospital because in the event of hospitalization, their health care providers opt to discharge them as soon as they are medically fit to be discharged given that most of them can be monitored remotely.  It could also be the case that the timely response by health care providers due to telecare use enables telecare users to avert serious health complications that would have seen them stay longer in hospital.  Because individuals have different biological endowments, which are typically unobservable, their health states and expected length of stay in hospital could vary even when using the same telecare devices.

We also argue that the consumption bundles of telecare users comprise several affordable health enhancing inputs in addition to telecare but telecare users choose telecare to improve their health status.  Accordingly, that telecare improves individuals' welfare by shortening their length of stay in hospital and also that it improves health implies that telecare use is an input in both the individuals' utility functions and health production functions.  This conceptualization makes our framework to have a considerable advantage over the experimental studies in the literature as it enables us to estimate the treatment effect whilst allowing individuals' health production to vary.  Furthermore, it is also appealing to economists and health planners who need to design policies based on measures of effectiveness that are obtained from real-world settings. 

Economists interested in conducting impact evaluations usually use economic theory to justify the causal inferences that they make and our particular framework is no different.  In addition to thinking of telecare use as a demand relation, in which case telecare users choose whether or not to use telecare devices based on prevailing market conditions, we also broadly classify the predictors of the length of stay in hospital as `predisposing', `enabling' and `need' factors following the Andersen's Behavioral Model of Health Services Use (see \cite{r6,r7} for a comprehensive discussion of these factors).  The predisposing factors are factors such as age, sex and area of residence that could make individuals to be more likely to have a longer length of stay in hospital (see, for example, \cite{r1} and \cite{r15} who find that older individuals have a longer length of stay in hospital than their younger counterparts on average; \cite{r26} who notes that rural residents are more likely to have a longer length of stay in hospital than the individuals who reside in urban areas, other factors held constant, and \cite{r33,r52} and \cite{r36} who note that males have a longer length of stay in hospital than females, holding other factors constant).  Enabling factors are factors such as individuals' income levels and access to health care services that enable or impede service use and could thus in turn have an impact on the length of stay in hospital.  Since we argue that telecare users may have a comparatively short length of stay in hospital by using telecare devices to substitute for some of the services provided by their care providers, it follows then that the use of telecare is an enabling factor in our case.  The need factors are factors such as multimorbidity and polypharmacy that are indicative of individuals' care needs.  We therefore expect that, all things equal, a high level of need is associated with a relatively longer length of stay in hospital (see, for example, \cite{r1}).

To set the scene for the analytical problem in this paper, consider a hypothetical population comprising several individuals who use health care services.  These individuals have a greater predisposition to use health care services to improve their health and well-being than their peers and can afford to purchase whatever services they want, but a number of them choose to use telecare instead of the other services for some reason.  We recognize that individuals have unique health endowments that directly affect their observable health status whether or not they use health care services and because of this heterogeneity, their likelihood of being hospitalized due to health deterioration may differ even when they all decide to use the same form of care.  We may also observe systematic variations in their length of stay in hospital  due to this heterogeneity and not necessarily because of their differences in use of services.  This is illustrated in Figure~\ref{Figure 1}.
\tikzstyle{decision}=[diamond,draw,fill=blue!50]
\tikzstyle{line}=[draw,-stealth,thick]
\tikzstyle{elli}=[draw,ellipse,fill=red!50]
\tikzstyle{block}=[draw,rectangle,rounded corners, text width=15em,text centered,minimum height=15mm,node distance=5em]
\begin{figure}[t!]
\centering
\begin{tikzpicture}
\node[block,text width=7em,yshift=-5em](Health and Social Care System){Health and Social Care System};
\node[block,text width=7em,left of=Health and Social Care System, xshift=-10em](Endowment){Endowment};
\node[block,text width=7em,below of=Health and Social Care System,yshift=-10em](Length of stay in hospital){Length of stay in hospital};
\node[block,text width=7em,below of=Length of stay in hospital,yshift=-2em](Health outcomes){Health outcomes};
\node[block,text width=7em,below of=Endowment,yshift=-10em](Preferences){Preferences};
\draw[dotted,thick] (-2.5,1) -- (-2.5,-9.5);
\draw[thick] (-5.0,-0.75) -- (-5.0,0);
\draw[thick] (-5.0,0) -- (5.0,0);
\draw[line, thick] (5.0,0) -- (5.0,-1.59);
\draw[thick] (-3.33,-6) -- (-2,-6);
\draw[thick] (-2,-6) -- (-2,-0.5);
\draw[thick] (-2,-0.5) -- (4.0,-0.5);
\draw[line, thick] (4.0,-0.5) -- (4.0,-1.59);
\draw[line, thick] (-3.33,-6) -- (-1.16,-6);
\draw[line, thick] (0,-2.25) -- (0,-5.22);
\draw[line, thick] (0,-6.73) -- (0,-7.30);
\draw[thick] (1.15,-8.5) -- (5.0,-8.5);
\draw[line, thick] (5.0,-8.5) -- (5.0,-5.57);
\draw[thick] (4.0,-5.57) -- (4.0,-7.5);
\draw[line, thick] (4.0,-7.5) -- (1.15,-7.5);
\draw[line, thick] (4.0,-6.0) -- (1.16,-6.0);
\draw[thick] (0.5,-5.22) -- (0.5,-3.5);
\draw[line, thick] (0.5,-3.5) -- (2.75,-3.5);
\draw[thick] (-5.65,-1.5) -- (-6.5,-1.5);
\draw[thick] (-6.5,-1.5) -- (-6.5,-8.0);
\draw[line, thick] (-6.5,-8.0) -- (-1.16,-8.0);
\node[text width=3cm, anchor=west, right] at (-5.75,0.5)
    {Unobserved factors};
\node[text width=3cm, anchor=east, left] at (6.5,0.5)
    {Observed factors};
\node[block,text width=9em,text centered,minimum height=40mm,right of=Health and Social Care System,xshift=9em,yshift=-7em](Observable factors){\begin{enumerate}
 \item Predisposing factors
    \item Enabling factors
    \item Need factors
\end{enumerate}};
\tikzstyle{line}=[draw,-stealth,thick]
\end{tikzpicture}
\caption{An illustration of individuals' utility maximizing behavior} \label{Figure 1}
\end{figure}
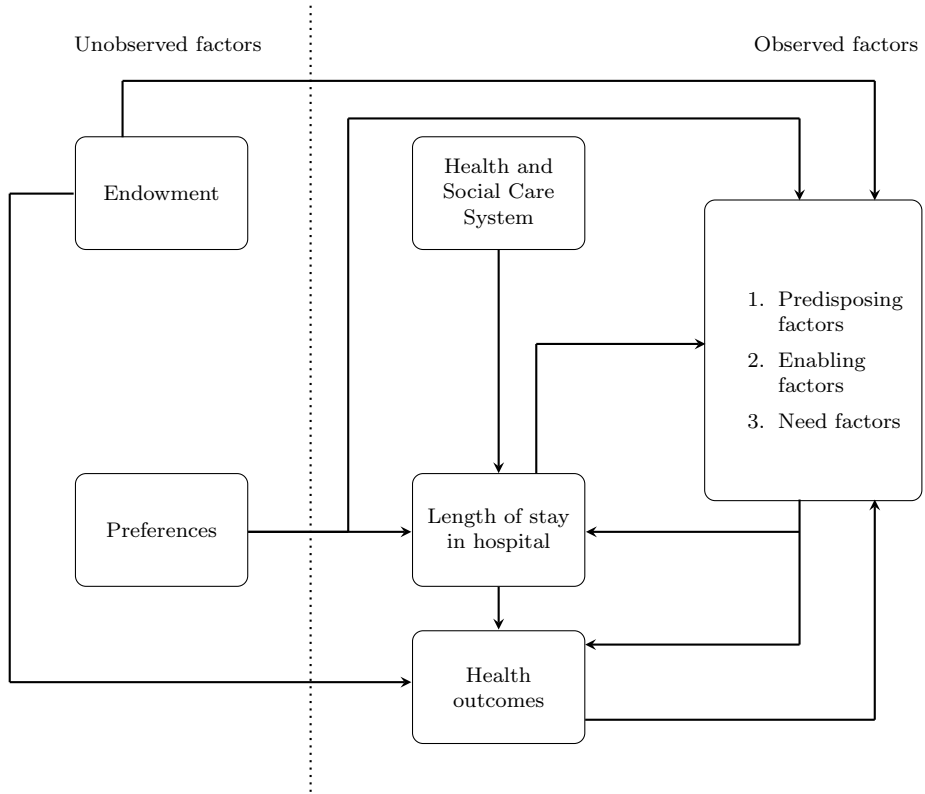

According to the figure, the length of stay in hospital is related to `predisposing', `enabling' and `need' factors.  The conceptual framework also shows that the observed variation in the length of stay in hospital may be brought about by factors operating at the health and social care system.  We can further observe from the figure that a particular individual's length of stay in hospital is influenced by the individual's level of endowment and preferences regarding the use of health care services, albeit indirectly.  The model developed in this paper causally links the decisions made by individuals to use telecare to their expected length of stay in hospital whilst disentangling the effects of unobservable confounding.

The majority of the previous studies have examined the effectiveness of telecare using experimental study designs.  A systematic review of the literature (\cite{r42}) showed that there are some studies that use randomized controlled trials (see, for example, \cite{r54,r29,r27}) and others that use quasi-experimental study designs such as the PSM technique and the Alternating Treatment Design (ATD) (see, for example, \cite{r18,r2,r4}).  Of these studies, it is only the studies by Akematsu and Tsuji and \cite{r54} that investigate the effect of telecare on the length of stay in hospital.  The studies by Akematsu and Tsuji find that telecare use leads to a decrease in the number of treatment days but their analytic sample is a potentially highly selected group of individuals, whereas the study by \cite{r54} does not find telecare use to be a significant predictor of the length of stay in hospital.  In order to encapsulate the mechanics of our proposed framework and, therefore, extend the literature, we use a unique dataset that contains non-experimental data from four different information sources.  

In particular, our unique dataset is a merger of the Scottish Morbidity Records (SMRs), prescribing data (which contains information on prescribed medications in Scotland), Self-Directed Support (SDS) data (which contains information on the decisions made by the individuals using social care services in Scotland regarding the provision of the services that they receive as well as data on several demographic characteristics) and the Homecare Census data for five local council areas in Scotland during the $2010/2011$ financial year.  The SMRs contain episode level data on acute hospital admissions (SMR01) and psychiatric care admissions (SMR04).  We construct our outcome measure of interest as the number of days that a particular individual spends in hospital while receiving treatment.  We also take advantage of the episode level data structure to generate time series data with repeated cross-sections over the $52$ weeks of the $2010/2011$ financial year.  We construct variables for age, sex, client group (which is a categorical variable with the following categories: `Dementia and Mental Health', `Learning disability', `Frail elderly' and `Physical disability'), telecare use, area of residence and comorbidity status; which serve as the covariates for the empirical analysis.  The variable for telecare use comprises the use of devices such as linked pill dispensers, linked smoke detectors, bogus caller buttons, property exit sensors and automated motion sensors, among others.

Our objective in this paper, therefore, is to formulate three econometric models that can be used to derive the causal impact of telecare on the length of stay in hospital in a real-world setting and compare the results with that of an experimental study design.  We specifically use the PSM technique for this demonstration due to its popularity in estimating causal effects using observational data.  The rest of the paper is as follows: Section~\ref{Section 2} introduces a theory of the demand for telecare; Section~\ref{Section 3} discusses the formulated econometric models; Section~\ref{Section 4} applies the formulated models to real-world data, and Section~\ref{Section 5} concludes.  An Appendix \ref{Appendix A} provides the proofs of our most important results.

\section{A theory of the demand for telecare}\label{Section 2}

\subsection{Notation and definitions}\label{subsection 2.1}

In order to operationalize the theoretical model for this study, we consider a constrained utility maximizing problem where individuals maximize utility and, by extension, their health status, $\mathcal{H}$, by using several utility generating inputs.  Let $\mathcal{D}\subset\mathbb{R}^n$ denote a set of utility generating inputs that an individual can consume at a particular point in time with notation $n$ in $\mathbb{R}^n$ indicating that individuals consume a finite amount of goods and services.  Let also $d\in \mathcal{D}$ represent individuals' non-negative consumption bundles such that $d=\mathbb{R}_+^n$ or more formally, $d\in\mathbb{R}^n:d\geq0$ for all $n=1,\ldots, N$.  Recall that since individuals simultaneously maximize utility and improve their health status, there exists $\mathcal{S}\in d$ such that $\mathcal{S}=\{\mathcal{X,K}\}$ and $\mathcal{S}\subseteq d$.  Here, $\mathcal{X}\in \mathcal{S}$ denotes the set of inputs that are directly associated with utility gains and $\mathcal{K}\in \mathcal{S}$ the set of health-related inputs that yield utility by improving $\mathcal{H}$.

If we let $\mathcal{B}=\{\mathcal{P,N}\}$, where $\mathcal{P}$ is a set of predisposing factors and $\mathcal{N}$ a set of need factors, so that we now have $\mathcal{X} \in \mathbb{R}^n$ such that $\mathcal{X}\geq0 \forall n=1,\ldots,N, \mathcal{K}\in \mathbb{R}_+^n$ and $\mathcal{B},d,\mathcal{S}\in \mathbb{R}^n$ such that $\mathcal{B},d,\mathcal{S}\geq0$ for all $n=1,\ldots,N$, then $\exists \mathcal{P}\in \mathcal{X}$ and also $\exists \mathcal{N}\in \mathcal{X}$ such that $\mathcal{P}=\{\mathcal{P}_1 + \mathcal{P}_2 +\ldots+ \mathcal{P}_n\}\subsetneq\mathcal{X}$, $\mathcal{N}=\{\mathcal{N}_1 + \mathcal{N}_2 +\ldots+ \mathcal{N}_n\}\subsetneq\mathcal{X}$ and $\mathcal{B}\subseteq\mathcal{X}$.  Because telecare use, $\mathcal{T}$, is a health-related enabling factor, we have $\mathcal{T,Q}\in \mathcal{K}$ such that $\mathcal{T}=\sum_{i=1}^{n} \mathcal{T}_i \subsetneq\mathcal{K}$ and $\mathcal{Q}=\sum_{i=1}^{n} \mathcal{Q}_i \subsetneq\mathcal{K}$.  There also exists $S\in \mathcal{S}$ such that $S=\{\mathcal{B,K}\}=\{\mathcal{P,N,T,Q}\}$ and $S\subset\mathcal{S}$.  The basic utility function for a particular individual is thus expressed as follows:
\begin{equation}
    U=u(S,\mathcal{H}) = u \left(\sum_{i=1}^{n} \mathcal{P}_i ,\{\mathcal{N}_1 +\ldots+ \mathcal{N}_n\} ,\{\mathcal{T}_1+\ldots+ \mathcal{T}_n\},\sum_{i=1}^{n} \mathcal{Q}_i ,\mathcal{H}\right)\label{eq 2.1}
\end{equation}

Following \cite{r65}, health seeking individuals demand health because good health is desirable and they thus invest in $\mathcal{T}$ and $\mathcal{Q}$ in order to produce their own health.  We can, therefore, express a particular individual's health production function at a certain point in time as $\mathcal{H}=H(\mathcal{T,Q},\mu)$, where $\mu$  denotes the unobservable biological endowment that affects the individual's health status.
\begin{remark}
Throughout this paper, the letters $P,N,T,Q,K$ and $H$ may also denote the utility generating inputs in Equation (\ref{eq 2.1}).  We, therefore, have that $P\in \mathcal{P}$, $N\subset \mathcal{N}$, $T\in \mathcal{T}$, $\mathcal{Q}\supset Q$, $\mathcal{K}\supset K$ and $H\in \mathcal{H}$.  Since $\mathcal{P}$, $\mathcal{N}$, $\mathcal{T}$, $\mathcal{Q}$, $\mathcal{K}$ and $\mathcal{H}$ contain more elements than $P,N,T,Q,K$ and $H$, one may think of $\mathcal{P}$, $\mathcal{N}$, $\mathcal{T}$, $\mathcal{Q}$ and $\mathcal{K}$ as sets of all the utility generating inputs that a particular individual can conceivably consume and $\mathcal{H}$ as a health production function with all the health enhancing inputs that the individual can conceivably acquire. 
\end{remark}
\begin{remark}
The sets $\mathcal{P}$, $\mathcal{N}$, $\mathcal{T}$ and $\mathcal{Q}$ are represented as sums of their constituent elements because individuals typically consume a combination of inputs.  We expect $U=u(.,.)$ to be additively separable in $\mathcal{T}$ and $\mathcal{Q}$ but not necessarily in $\mathcal{P}$ and $\mathcal{N}$.  Since $\mathcal{P}$, $\mathcal{N}$, $\mathcal{T}$ and $\mathcal{Q}$ are also disjoint sets, we have that $\mathcal{T}\cup (\mathcal{P}\cap \mathcal{N})=(\mathcal{T}\cup \mathcal{P})\cap (\mathcal{T}\cup \mathcal{N})=\mathcal{T}=\mathcal{K}\triangle \mathcal{Q}=\mathcal{K}\cup \mathcal{Q}\setminus \mathcal{K}\cap \mathcal{Q}$ and that $\chi (\mathcal{K}\triangle \mathcal{Q})=\chi_{\mathcal{K}}\oplus \chi_{\mathcal{Q}}=\mathcal{T}$.  It is also the case that $[\mathcal{T}\in \mathcal{K}\triangle \mathcal{Q}]=[\mathcal{T}\in \mathcal{K}]\oplus [\mathcal{T}\in \mathcal{Q}]$ and that $\mathcal{T}\cap (\mathcal{P}\cup \mathcal{N})=(\mathcal{T}\cap \mathcal{P})\cup (\mathcal{T}\cap \mathcal{N})=\left\{  \right\}$.  We also have that $\bigcup_{\mathcal{K,B}\in S} S=\mathcal{P}+\mathcal{N}+\mathcal{T}+\mathcal{Q}$ since $\exists$ at least one utility generating input $l\in S$ such that $l$ is a member of at least one of $\mathcal{P,N,T}$ and $\mathcal{Q}$ .
\end{remark}
\begin{definition}\label{def1}
Suppose that $A$ and $M$ are such that $A\subset \mathbb{Z}$ and $M\subset \mathbb{Z}$.  Then $A$ and $M$ are said to be disjoint sets if $A\cap M\Leftrightarrow \emptyset$.
\end{definition}
In addition to the constraints that we have imposed on individuals' consumption bundles in that $S\in \mathbb{R}_{\textgreater 0}^n$, individuals are also limited to the inputs that they can afford.  A feasible consumption bundle for a particular individual, therefore, is that which the individual is able to acquire given the individual's budget.
\begin{definition}\label{def2}
Let $S^{'}$ denote the set of utility generating inputs in $S$ that are also part of $\mathcal{P,N,T}$ and $\mathcal{Q}$ depending on a particular individual's consumption pattern.  Let also $w$ denote the individual's wealth level and $p$ the vector of prices for inputs $S^{'}$.  Then the individual's feasible consumption bundle consists of the utility generating inputs in the set $\{S^{'}\in S:p.S^{'}\leq w\}$.  For a market economy where the consumers are price takers, this set is the so called Walrasian budget set (see, for example, \cite{r12} and \cite{r55}).  The corresponding budget constraint can, therefore, be written as $p_{1}.S_{1}^{'}+\ldots+p_{n}.S_{n}^{'}\leq w$.
\end{definition}
\subsection{Model assumptions}\label{subsection 2.2}
We make three key assumptions that enable us to conduct our analysis.  These assumptions provide a theoretical basis upon which we can make causal inferences about the expected relationship between the demand for telecare and our outcome measure of interest.
\begin{asu}[Implicit input prices]\label{Assumption 1}

For all $\{P,N,T,Q\}\in S^{'}$, there exists $P=\{p_P , p_N , p_T , p_Q\}$ such that $p_P\in P, p_N\in N, p_T\in T$ and $p_Q\in Q$.

\end{asu}
Assumption \ref{Assumption 1} states that the utility generating input prices are implicit.  We make this assumption because there are some utility generating inputs in our modeling framework such as $\mathcal{T}$ and $\mathcal{Q}$ that are traded in the market place and others such as $\mathcal{P}$ and $\mathcal{N}$ that do not have natural markets.  We also do not have data on the actual input prices for $\mathcal{T}$ and $\mathcal{Q}$ but we know that individuals usually incur costs (which may include opportunity costs and other intangible costs) in the course of acquiring inputs.  Even though $\mathcal{P}$ and $\mathcal{N}$ are not marketed, we can still be able to rigorously think about their cost implications.  For instance, an elderly individual with a chronic condition is expected to have a comparatively high use of health and social care resources, all else equal.  The difference in resource use between the elderly individual with a chronic condition and another individual who is neither elderly nor suffering from a chronic condition can be thought of as the price associated with $\mathcal{P}$ and $\mathcal{N}$.  Furthermore, the fact that we are able to observe $S^{'}$ means that the input costs have already been incurred.  The implication of this assumption, therefore, is that we are able to conduct our analysis in the absence of observable prices and still have theoretical validity.
\begin{asu}[Weak axiom of revealed preference]\label{Assumption 2}
The choice structure $C(\mathcal{M})$ for a particular individual who uses $T$ is such that if for some budget set $M\in\mathcal{M}$ and $M\in S^{'}$ with $T,Q\in M$ and the individual chooses $T$ such that $T\in C(M)$, then $\forall M^{'}\in\mathcal{M}$ with $T,Q\in M^{'}$ and $Q\in C(M^{'})$, $\bigvee_{m}\in C(M^{'})$ that is $T$ so that we also have $T\in C(M^{'})$.  Suppose that $\exists\mathcal{Y}\subset S$ and $\mathcal{Y}\not\subset K$ but $S^{'}\in\mathcal{S}$ now contains $\mathcal{Y}$.  Then we may have that (i) $T\succ Q\Leftrightarrow T\succcurlyeq Q$ but not $Q\succcurlyeq T$; (ii) $Q\succ \mathcal{Y}\Leftrightarrow Q\succcurlyeq\mathcal{Y}$ but not $\mathcal{Y}\succcurlyeq Q$, and (iii) $\mathcal{Y}\succ T\Leftrightarrow\mathcal{Y}\succcurlyeq T$ but not $T\succcurlyeq\mathcal{Y}$.
\end{asu}
Assumption \ref{Assumption 2} states that a telecare user prefers to use $\mathcal{T}$ instead of $\mathcal{Q}$ whenever presented with a feasible consumption bundle that contains both $\mathcal{T}$ and $\mathcal{Q}$.  So in the event that input $\mathcal{Q}$ is chosen, it must be the case that $\mathcal{T}$ is also chosen.  This assumption is rather conservative as it does not require an individual's preference relation $\succcurlyeq$ to be transitive as is the case with the strong and generalized axioms of revealed preference (\cite{r51,r30,r59}).  Therefore, if an alternative set of utility generating inputs $\mathcal{Y}$ is to be considered such that $\mathcal{Y}$ is not part of either $\mathcal{T}$ or $\mathcal{Q}$, then a situation may arise where (i) telecare use is at least as preferable as $\mathcal{Q}$ but $\mathcal{Q}$ is not at least as preferable as telecare use; (ii) $\mathcal{Q}$ is at least as preferable as $\mathcal{Y}$ but $\mathcal{Y}$ is not at least as preferable as $\mathcal{Q}$, and (iii) $\mathcal{Y}$ is at least as preferable as telecare use but telecare use is not at least as preferable as $\mathcal{Y}$.  By making this assumption, we acknowledge that individuals' preferences could change over time given that their biological endowments are typically unobservable and it may also no longer be possible to slow down the expected deterioration of $\mathcal{H}$ with the same set of inputs that were previously efficacious.  Assumption \ref{Assumption 2} is thus the reasonable stance to take as it is not uncommon for health seeking individuals--particularly those suffering from chronic conditions such as cancer--with a predefined ranking of preferences for health inputs to resort to some form of palliative care at a later point in time having lost faith in the effectiveness of their earlier treatment choices.
\begin{asu}[Wealth equivalence]\label{Assumption 3}
Any feasible consumption bundle $S^{'}$ that is represented by $S^{'}(p,w)$ price wealth situation must be such that $w=w^{'}$ for all $w\in\mathcal{W}$, ceteris paribus, if $\Vert U_{i}:\mathbb{R}^{\mathrm{L}}_{+}\times\mathbb{R}\rightarrow\mathbb{R} \|_{\infty}=\varphi$ for all $i=i,\ldots,n$.  If this is not the case, it follows then that
\begin{enumerate}[label=\roman*.]
    \item Since $p.S^{'}\leq w,\exists$ some $S^{'}(p,w)\in\mathcal{S}$ where $w^{'}\textgreater w$ and thus $\sup\{\vert f(S^{'},H) \vert\}$ where $S^{'}(p,w^{'})\in\mathcal{S}$ is greater than $\bigvee_{ S,H\in U}\{U:S(p,w)\}$, all else equal.
    \item There could be instances where $p_{Q}\neq p_{T}$.   
\end{enumerate}
\end{asu}
Assumption \ref{Assumption 3} states that the wealth levels for all individuals have to be the same if we are to attribute utility gains to the consumption of utility generating inputs.  This assumption is useful since we are interested in deriving the maximal utility that is gotten from the consumption of $\mathcal{T}$.  Because the budget constraint is represented by the inequality $\leq w$, we know that the utility functions of the less affluent individuals yield lower utility than those of their more affluent counterparts, all else equal, due to the fact that they have to deplete their wealth to acquire a given consumption bundle.  In order to eliminate such instances, therefore, the analysis is limited to the wealth levels that are just sufficient to purchase the utility generating inputs i.e. the budget hyperplane or equivalently the budget line when considering only two inputs in the consumption bundle.  This is demonstrated in Figure \ref{Figure 2} for a hypothetical case of $K_{1}$ and $\mathcal{B}_{2}$ as the utility generating inputs. 
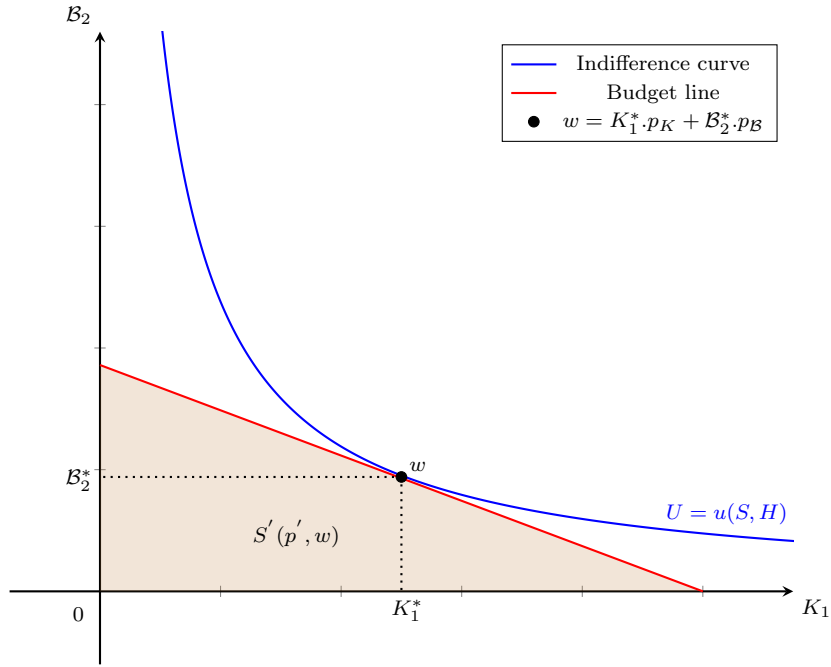
\begin{figure}[t!]
\centering
\begin{tikzpicture}
\begin{axis}[standard,
width=12cm, 
height=10cm,
xlabel={$K_{1}$},
ylabel={$\mathcal{B}_{2}$},
xticklabel=\empty,
yticklabel=\empty,
samples=1000,
xmin=0,  xmax=5,
ymin=0, ymax=20]

\node[anchor=center, label=south west: $0$] at (axis cs:0,0){};
\addplot [
        thick,
        blue,
        domain=0:9,
    ] {11.9/x};
 \addlegendentry{Indifference curve}

\node[text width=2cm, anchor=west, right] at (axis cs:1.2,1.5) [anchor=south west] {$S^{'}(p^{'},w)$};
\node [color=blue] at (5.2,3.2) {$U=u(S,H)$};
\addplot[
        name path=F,
        domain=0.00:5,
        color=red,
        thick,
        samples=1000
]
    {9.3 - (9.3/5)*x};
\addlegendentry{Budget line}

 \addplot[
        only marks,
        mark=*,
        mark size=2pt,
        color=black,
    ] coordinates {(2.5,4.7)};  
\addlegendentry{$w=K_{1}^{*}.p_{K}+\mathcal{B}_{2}^{*}.p_{\mathcal{B}}$}
\path[name path=xAxis] (axis cs:0,0) -- (axis cs:5,0);
\addplot[fill=brown, fill opacity=0.2] fill between [of=F and xAxis];
\draw[dotted,thick] (2.5,4.7) -- (2.5,0);
\node[text width=2cm, anchor=west, right] at (axis cs:2.35,-1.55) [anchor=south west] {$K_{1}^{*}$};
\draw[dotted,thick] (2.5,4.7) -- (0,4.7);
\node[text width=2cm, anchor=west, right] at (axis cs:-0.35,3.8) [anchor=south west] {$\mathcal{B}_{2}^{*}$};
\node[text width=2cm, anchor=west, right] at (axis cs:2.5,4.6) [anchor=south west] {$w$};

\end{axis}
\end{tikzpicture}
\caption{Individuals' utility maximizing consumption bundle at point $w$} \label{Figure 2}
\end{figure}
According to the figure, the feasible consumption bundle that maximizes utility comprises the quantities of $K_{1}$ and $\mathcal{B}_{2}$ at the point of tangency between a particular individual's indifference curve and the individual's budget constraint.  Since the analysis focuses only on the individuals who have just enough resources needed to acquire the utility generating inputs and we also know that telecare users choose $\mathcal{T}$ instead of $\mathcal{Q}$ to be part of their utility functions, it follows then that the price of $\mathcal{Q}$, denoted by $p_{Q}$, is the relative price of $\mathcal{T}$, holding other factors constant.  Accordingly, a particular individual who is a telecare user chooses to use telecare strictly because of preference and not because it costs less than its substitutes.

We now proceed to present Theorem \ref{Theorem 1} having made the assumptions needed to identify the demand for telecare.  The theorem shows how one can derive the demand function for telecare given a particular individual's budget constraint and the individual's health production function.  We show that the demand for telecare depends in part on an individual's unobservable biological endowment. 
\begin{theorem}\label{Theorem 1}

Suppose Assumptions $1-3$ all hold and let $p_{\mathcal{B}Q}\subset p$ such that $p_{\mathcal{B}Q}=\{p_{P},p_{N},p_{Q}\}$.  Then for any utility maximizing behavior that is characterized by $U=u(S_{\mathcal{B}Q}^{'},H)$ with $p=\{p_{\mathcal{B}},p_{K}\}$ as the input prices, the demand function for $T$ will be given by $T=t(p,w,\mu)$ and the indirect utility function by $V=v(\mathcal{B},T,p_{\mathcal{B}Q},w,\mu)$.

\end{theorem}
\begin{proof}
By Assumptions (\ref{Assumption 1}) and (\ref{Assumption 3}), $U=u(.,.)$ can be maximized $\text{s.t}$ $\{p.S^{'}=w=p_{1}.S_{1}^{'}+\ldots+p_{n}.S_{n}^{'}\}$ and $H=H(T,Q,\mu)$.  Recall that since $S^{'}$ is defined as $S^{'}\in S$ such that $S=\{\mathcal{B},K\}$ with $P,N\in\mathcal{B}$ and $K=\{T,Q\}$, we can also write the optimization problem as
\[
\max_{\{\mathcal{B},T\}\in S,H} \quad U\left(P,N,T,H\right)\hspace*{0.5em} \text{s.t} \hspace*{0.5em}w=p_{P}P+p_{N}N+p_{T}T+p_{Q}Q\hspace*{0.5em} \text{and} \hspace*{0.5em}H=H\left(T,Q,\mu\right)
\]
Using the Lagrange method, the Augmented Objective Function (AOF) for the optimization problem is given by
\begin{eqnarray*}
  \mathcal{L}\left(\mathcal{B},T,H:\lambda,\lambda_{2}\right) & = & U(.)+\lambda \left(p_{P}P+p_{N}N+p_{T}T+p_{Q}Q-w\right)+\lambda_{2}H(.)\\
 & = &   U(.)+\lambda p_{P}P+\lambda p_{N}N+\lambda p_{T}T+\lambda p_{Q}Q-\lambda w+ \lambda_{2}H(.)\\
  & = &   U(.)+\left(\lambda p_{P}P+\lambda p_{N}N+\lambda p_{T}T+\lambda p_{Q}Q-\lambda w\right)+ \lambda_{2}H(.)\\
\end{eqnarray*}
Note that the AOF is such that $\mathcal{L}(.,.,.:\lambda,\lambda_{2})$ has the utility function $U=U(.)$, four price components of $P,N,T$ and $Q$, and components $\lambda w$ and $\lambda_{2}H(.)$.  We obtain the first order conditions needed to get $\sup\{U:\mathbb{R}_+^n\rightarrow\mathbb{R}\}$ by computing ($\mathrm{d}\mathcal{L}/\mathrm{d}/{T}),(\mathrm{d}\mathcal{L}/\mathrm{d}\lambda)$ and $(\mathrm{d}\mathcal{L}/\mathrm{d}\lambda_{2})$ and then equating the differentials to zero.  For the case of ($\mathrm{d}\mathcal{L}/\mathrm{d}/{T}$) we, therefore, have that
\begin{eqnarray*}
  \diff{\mathcal{L}}{T} &=&  \diffp{U}{T}+\diffp{U}{N}\times\diff{N}{T}+\diffp{U}{P}\times\diff{P}{T}+\diffp{U}{H}\times\diff{H}{T}\\
 &+&\diff{\left(\lambda p_{T}T\right)}{T}+\diffp{\lambda_{2}H(.)}{T}+\diffp{\lambda_{2}H(.)}{Q}\times\diff{Q}{T}+\diffp{\lambda_{2}H(.)}{\mu}\times\diff{\mu}{T}
\end{eqnarray*}
\begin{eqnarray*}
\diff{\mathcal{L}}{T} &=&  \diffp{U}{T}+\diffp{U}{N}\times\frac{\mathrm{d}}{\mathrm{d}T}\left(\int_{T}\diff{N}{T} \,dT\right)+\diffp{U}{P}\times\frac{\mathrm{d}}{\mathrm{d}T}\left(\int_{T}\diff{P}{T} \,dT\right)\\
&+&\diffp{U}{H}\times\frac{\mathrm{d}}{\mathrm{d}T}\left(\int_{T}\diff{H}{T} \,dT\right)+\diffp{(\lambda p_{T}T)}{T}\times\diff{T}{T}+\diffp{\lambda_{2}H(.)}{T}\\
&+&\diffp{\lambda_{2}H(.)}{Q}\times\frac{\mathrm{d}}{\mathrm{d}T}\left(\int_{T}\diff{Q}{T} \,dT\right)+\mu
\end{eqnarray*}
\begin{eqnarray*}
\diff{\mathcal{L}}{T} &=&  \diffp{U}{T}+\diffp{U}{N}\times\frac{\mathrm{d}}{\mathrm{d}T}\left(\int_{T}\diff{N}{T} \,dT\right)+\diffp{U}{P}\times\frac{\mathrm{d}}{\mathrm{d}T}\left(\int_{T}\diff{P}{T} \,dT\right)\\
&+&\diffp{U}{H}\times\frac{\mathrm{d}}{\mathrm{d}T}\left(\int_{T}\diff{H}{T} \,dT\right)+\diffp{(\lambda p_{T}T)}{T}+\diffp{\lambda_{2}H(.)}{T}\\
&+&\diffp{\lambda_{2}H(.)}{Q}\times\frac{\mathrm{d}}{\mathrm{d}T}\left(\int_{T}\diff{Q}{T} \,dT\right)+\mu
\end{eqnarray*}
Suppose that $\exists\zeta\subset\mathbb{R}$ such that $\zeta$ contains $(\partial U/\partial T),(\partial U/\partial N),(\partial U/\partial P)$ and $(\partial U/\partial H)$ where $(\partial (.)/\partial (.))\in\zeta$ are the partial derivatives of $U=U(.)$ with respect to $P,N,T$ and $H$, and also that $\exists\mathcal{J}=\{(\partial H/\partial T),(\partial H/\partial Q)\}$ where $(\partial (.)/\partial (.))\in\mathcal{J}$ are the partial derivatives of $H=H(.)$ with respect to $T$ and $Q$.  Then $\mathrm{d}\mathcal{L}/\mathrm{d}/{T}$ can be further simplified as
\begin{eqnarray*}
\diff{\mathcal{L}}{T} &=&  \zeta_{T}+\zeta_{N}\times\frac{\mathrm{d}}{\mathrm{d}T}\left(\int_{T}\diff{N}{T} \,dT\right)+\zeta_{P}\times\frac{\mathrm{d}}{\mathrm{d}T}\left(\int_{T}\diff{P}{T} \,dT\right)\\
&+&\zeta_{H}\times\frac{\mathrm{d}}{\mathrm{d}T}\left(\int_{T}\diff{H}{T} \,dT\right)+\diffp{(\lambda p_{T}T)}{T}+\lambda_{2}\mathcal{J}_{T}\\
&+&\lambda_{2}\mathcal{J}_{Q}\times\frac{\mathrm{d}}{\mathrm{d}T}\left(\int_{T}\diff{Q}{T} \,dT\right)+\mu
\end{eqnarray*}
which can also be expressed as
\begin{eqnarray*}
\diff{\mathcal{L}}{T} &=&  \zeta_{T}+\zeta_{N}\times\frac{\mathrm{d}}{\mathrm{d}T}\left(\int_{T}\diff{N}{T} \,dT\right)+\zeta_{P}\times\frac{\mathrm{d}}{\mathrm{d}T}\left(\int_{T}\diff{P}{T} \,dT\right)\\
&+&\zeta_{H}\times\frac{\mathrm{d}}{\mathrm{d}T}\left(\int_{T}\diff{H}{T} \,dT\right)+\lambda p_{T}+\lambda_{2}\mathcal{J}_{T}\\
&+&\lambda_{2}\mathcal{J}_{Q}\times\frac{\mathrm{d}}{\mathrm{d}T}\left(\int_{T}\diff{Q}{T} \,dT\right)+\mu
\end{eqnarray*}
so that $\mathrm{d}\mathcal{L}$ is given by
\begin{eqnarray*}
\mathrm{d}\mathcal{L} &=&  \zeta_{T}\mathrm{d}T+\zeta_{N}\times\mathrm{d}\left(\int_{T}\diff{N}{T} \,dT\right)+\zeta_{P}\times\mathrm{d}\left(\int_{T}\diff{P}{T} \,dT\right)\\
&+&\zeta_{H}\times\mathrm{d}\left(\int_{T}\diff{H}{T} \,dT\right)+\lambda p_{T}\mathrm{d}T+\lambda_{2}\mathcal{J}_{T}\mathrm{d}T\\
&+&\lambda_{2}\mathcal{J}_{Q}\times\mathrm{d}\left(\int_{T}\diff{Q}{T} \,dT\right)+\mu
\end{eqnarray*}
Following the same logic that was used to derive $\mathrm{d}\mathcal{L}$ by first computing ($\mathrm{d}\mathcal{L}/\mathrm{d}/{T}$), we can also derive $\mathrm{d}\mathcal{L}$ by computing ($\mathrm{d}\mathcal{L}/\mathrm{d}\lambda$) and ($\mathrm{d}\mathcal{L}/\mathrm{d}\lambda_{2}$) as
\begin{eqnarray*}
\mathrm{d}\mathcal{L} &=& \left[p_{P}\left(\int_{\lambda}\diff{P}{\lambda} \,d\lambda\right)\right]\mathrm{d}\lambda                                                                                                       +\left[p_{N}\left(\int_{\lambda}\diff{N}{\lambda} \,d\lambda\right)\right]\mathrm{d}\lambda+\left[p_{T}\left(\int_{\lambda}\diff{T}{\lambda} \,d\lambda\right)\right]\mathrm{d}\lambda\\
&+& \left[p_{Q}\left(\int_{\lambda}\diff{Q}{\lambda} \,d\lambda\right)\right]\mathrm{d}\lambda-w\mathrm{d}\lambda\\
\text{and}\\
\\
 \mathrm{d}\mathcal{L} &=& H\left(.\right)\mathrm{d}\lambda_{2}                                                                                                      
\end{eqnarray*}
By equating $\mathrm{d}\mathcal{L}$ to zero, we can write the demand equation for $T$ by solving for the integrals and writing $T$ on the left-hand side as follows:
\begin{equation}
T= \mathlarger{\mathlarger{\sum}}_{\mathcal{B},T,H\in\mathcal{A}}\frac{\zeta_{\mathcal{A}}\mathrm{d}\mathcal{A}}{p_{T}\mathrm{d}\lambda}+\mathlarger{\mathlarger{\sum}}_{T,Q\in K}\frac{\lambda_{2}\mathcal{J}_{K}\mathrm{d}K}{p_{T}\mathrm{d}\lambda}-\mathlarger{\mathlarger{\sum}}_{\mathcal{B},Q\in S^{'}}\frac{p_{\mathcal{B}Q} S_{\mathcal{B}Q}^{'}}{p_{T}}+\frac{\lambda\mathrm{d}T}{\mathrm{d}\lambda}+\frac{w}{p_{T}}+\mu\label{eq 2.2}       
\end{equation}
Equation \ref{eq 2.2} shows that the demand for $T$ for a particular individual can be arrived at by computing (i) the maximal utility that results from $T,N$ and $P$ when the budget constraints are relaxed i.e. $[(\partial U/\partial T)/p_{T}\times(\mathrm{d}T/\mathrm{d}\lambda)]+[(\partial U/\partial N)/p_{T}\times(\mathrm{d}N/\mathrm{d}\lambda)]+[(\partial U/\partial P)/p_{T}\times(\mathrm{d}P/\mathrm{d}\lambda)]+[(\partial U/\partial H)/p_{T}\times(\mathrm{d}H/\mathrm{d}\lambda)]$; (ii) the maximal utility that results from the consumption of $T$ and $Q$ when the budget constraints are relaxed and when the individual has become more efficient at producing health i.e. $[\lambda_{2}(\partial H/\partial T)/p_{T}\times(\mathrm{d}T/\mathrm{d}\lambda)]+[\lambda_{2}(\partial H/\partial Q)/p_{T}\times(\mathrm{d}Q/\mathrm{d}\lambda)]$; (iii) the quantity of $T$ that would be demanded if all the resources earmarked for acquiring $N,P$ and $Q$ were to be used to acquire $T$ i.e. $[(p_{N}N/p_{T})+(p_{P}P/p_{T})+(p_{Q}Q/p_{T})]$; (iv) the maximal utility that would result from $T$ due to a unit relaxation of the budget constraints i.e. $(\lambda\mathrm{d}T/\mathrm{d}\lambda)$; (v) the quantity of $T$ that would be demanded if the individual's entire wealth were to be used to acquire $T$ i.e. $w/p_{T}$, and the individual's unobservable biological endowment.
\begin{lemma}\label{lemma 1}
Given Assumptions (\ref{Assumption 1}), (\ref{Assumption 2}) and (\ref{Assumption 3}), for as long as $T$ assumes the general form specified in Equation \ref{eq 2.2}, then (i) $T\propto (1/p_{T},p_{Q})$, $T\propto p_{N}$, $T\propto p_{P}$, $T\propto\mu$ and $T\propto w$ must always be true, and (ii) $T$ and $Q$ are normal.
\end{lemma}
Lemma \ref{lemma 1} allows for inferences to be made regarding the expected relationships between the quantity of $T$ demanded, the input prices, the unobservable health production technology and individuals' wealth levels following economic theory and mathematical intuition.  By Lemma \ref{lemma 1}, we can express $T$ in Equation \ref{eq 2.2} as $T=t(p,w,\mu)$.  This completes the first part of the proof.

The input demand functions for $P,N$ and $Q$ can also be derived in the same fashion as the derivation of $T=t(.,.,.)$ to yield $P=r(p,w,\mu)$ as the demand function for $P$, $N=n(p,w,\mu)$ as the demand function for $N$ and $Q=q(p,w,\mu)$ as the demand function for $Q$.  By Assumption \ref{Assumption 2}, if $T\succcurlyeq Q$ then $T$ is chosen $\forall K^{'}(p,w)\in\mathcal{S}$ where $K^{'}=\{T,Q\}$.  If $T\succcurlyeq Q$, then by Assumption \ref{Assumption 3} we also know that $K^{'}(p_{Q},w)\equiv K^{'}(p_{T},w)$ since $w=\overline{w}$ for all users of $T$ and $Q$.  We can, therefore, formulate the indirect utility function $V=v(\mathcal{B},T,p_{\mathcal{B}Q},w,\mu)$ by first substituting $Q=q(.)$ into $H=H(T,Q,\mu)$ and further substituting the resultant function into $U=U(.)$.  This completes the second part of the proof.
\end{proof}
\begin{corollary}\label{Corollary 1}
 
Suppose that Assumptions (\ref{Assumption 1}), (\ref{Assumption 2}) and (\ref{Assumption 3}) hold and let the empirical specification of $V=v(.)$ be $Y=f(x)$.  Let also 
\begin{enumerate}
\item $\lim\limits_{\vert X \vert \to \infty} Y=E(Y\vert X)$
\item $\lim\limits_{\vert X \vert \to \infty} \hat{Y}=E(\hat{Y}\vert X)$
\end{enumerate}
Then $\Vert E(Y\vert X)-E(\hat{Y}\vert X) \|=\delta$ such that (i) $\delta=\{\epsilon,\mu\}$ and (ii) $\epsilon \perp \mu$.
 
\end{corollary}
Corollary \ref{Corollary 1} points out, loosely speaking, that the empirical version of the indirect utility function does not sufficiently explain the utility maximizing behavior of a particular individual since part of the individual's health production technology is unobservable.  A closer look at the indirect utility function and health production function also shows that an endogeneity problem could arise in the course of estimating the impact of the demand for telecare on the length of stay in hospital since $T$ is contained in both functions.  We suspect that telecare use is potentially endogenous as individuals typically choose whether or not to use telecare and, therefore, the unobservable factors guiding their choices might be correlated with the length of stay in hospital.  Notice from Equation \ref{eq 2.1} that because a particular individual's utility is directly related to the individual's level of endowment and that a telecare user chooses to use telecare, a problem of unobserved heterogeneity could arise if the differences in endowments and preferences within the study population cause the effect of telecare on the length of stay in hospital to differ among various individuals.
\section{Econometric model}\label{Section 3}

\subsection{Model setup}\label{subsection 3.1}
We estimate the treatment effect using a count regression model since the variable for the length of stay in hospital, $Y$, is a count variable (\cite{r37,r41,r9}).  In particular, we have that $Y_{it}\in\mathbb{Z}_{0}^{+}$ where the weekly cross-sections over the $2010/2011$ financial year are indexed by $t\in(1,\ldots,52)$ and $i$ denotes the cross-section unit at time $t$.  We assume that $Y_{it}$ is generated by an unobservable Poisson process that is given by the following equation:
\begin{equation}
Pr(Y_{it}=\mathrm{v})=\frac{\exp\left(-\lambda\right)\lambda^{\mathrm{v}}}{\mathrm{v!}}\label{eq 3.1}
\end{equation}
where $\lambda$ is the Poisson parameter being estimated and $\mathrm{v}$ is the observable count.

Suppose that some of our covariates of interest i.e. age, sex, client group, rurality and comorbidity status are represented by $\vec{X}$ such that $X=\{X_{i},\ldots,X_{n}\}$ where $n=5$ are $\mathcal{P,N}\in\mathcal{S}$ and $X\in\operatorname{supp}(X)$.  Suppose furthermore that the variable for telecare use, $T\in(0,1)$, represents $T\in\mathcal{K}$ and $T$ is such that $\forall$ values of $T\in\mathbb{Z}_{2}:\mathbb{Z}_{2}=\{0,1\}$, we have that $f(T)\textgreater0$.  Because the outcome variable, $Y_{it}$, is such that $Y=\{0,1,2,\ldots,n\}\in\mathbb{N}^{0}$ for all $i$ and $t$ with $Y_{it}\in\operatorname{supp}(Y_{it})$, we can write the basic count regression model relating $X=\{X_{jit}\}_{j=1}^{5}$ and $T_{it}$ to $Y_{it}$ as shown in Equations (\ref{eq 3.2}) and (\ref{eq 3.3}).
\begin{eqnarray}
  \lambda &=&\exp\left(\beta_0+\sum_{j=1}^{5}\beta_{j} X_{jit}+\omega T_{it}\right)\label{eq 3.2}\\
  \log\lambda &=& \beta_0+\sum_{j=1}^{5}\beta_{j} X_{jit}+\omega T_{it}\label{eq 3.3}
\end{eqnarray}
where $\beta$ and $\omega$ are the coefficients being estimated.

Notice from Equation \ref{eq 3.1} that since $Y_{it}\in\mathbb{Z}^{*}$, the specification of $Pr(.=.)$ with $\exp(.)$ ensures that $Pr(Y_{it}=\mathrm{v})\neq0$.  This is because for all $y\in Y$, it must be the case that $\exp(.)\textgreater0$.  A special characteristic of the model is that it assumes equidispersion (\cite{r32,r62,r25,r9}) i.e. $E(Y_{it}\vert X_{it}T_{it})=\lambda=\operatorname{Var}(Y_{it}\vert X_{it}T_{it})$.  To put it simply, the expected value of $Y_{it}$ is equal to its variance.

Following \cite{r38} and \cite{r39}, if we let $\lambda$ be a gamma random variable $\lambda\sim\Gamma(s,\theta)$ with a shape parameter, $s$, and a scale parameter, $\theta$, then $\lambda=\exp(.)$ in Equation \ref{eq 3.2} and $\log\lambda=(.)$ in Equation \ref{eq 3.3} would become Negative Binomial Models with $\operatorname{Var}(Y_{it}\vert X_{it},T_{it})=\lambda(1+\alpha\lambda)$ and $E(Y_{it}\vert X_{it},T_{it})=\lambda$; implying that there is overdispersion in the data since $\operatorname{Var}(Y_{it}\vert X_{it},T_{it})$ is greater than $E(Y_{it}\vert X_{it},T_{it})$.  Unlike the Poisson Model where $\operatorname{Var}(Y_{it}\vert.)=E(Y_{it}\vert.)$ for all $i$ and $t$, the Negative Binomial Model is less restrictive as it allows for the variance of $Y_{it}$ to exceed its expected value.  Furthermore, since the variance of $Y_{it}$ is expressed as $\lambda(1+\alpha\lambda)$ in the Negative Binomial Model formulation, the Poisson Model would be nested as a special case, specifically when $\alpha=0$.  The Negative Binomial Model is given by
\begin{eqnarray}
\log r=\beta_0+\mathlarger{\mathlarger{\sum}}_{j=1}^{5}\beta_{j} X_{jit}+\omega T_{it}\label{eq 3.4}
\end{eqnarray}
where $r$ is the estimated parameter of interest and $r\sim\operatorname{poisson} (\theta)$ such that $\theta\sim\operatorname{gamma} (r,\frac{Pr}{1-Pr})$.

Various versions of the Negative Binomial Model may be formulated depending on the nature of overdispersion.  For example, one may simply use a basic Negative Binomial Model and rely on its statistical properties to handle the underlying overdispersion.  Alternatively, one may use a zero--inflated Negative Binomial Model when $Y=\{0,1,2,\ldots,n\}\in\mathbb{Z}^{n}$ for all $i$ and $t$ contains more zeros than would normally be expected of a gamma--poisson mixed distribution (\cite{r16,r10}) or a zero--truncated Negative Binomial Model when the outcome measure contains only positive integers (\cite{r25}).

In this paper, we present the empirical results of the three formulations and compare the treatment effects with that derived from the \hyperref[PSM]{PSM technique}.  The first Model, Model \ref{Model 1}, is a basic Negative Binomial Model; the second model, Model \ref{Model 2}, is a zero--inflated Negative Binomial Model, whereas the third model, Model \ref{Model 3}, is a zero--truncated Negative Binomial Model.
\begin{model}\label{Model 1}
We let the observable count, $\mathrm{v}$, be represented as a function of $X$ and $T$ in a modeling framework such that $\mathrm{v}$ is determined by a linear predictor, $\eta$, which is equal to $f(X,T)$.  We let also $r$ be the unobservable process generating $Y_{it}$ such that $r\sim\operatorname{poisson}(\theta)$ and $\theta\sim\operatorname{gamma} (r,Pr/1-Pr)$.  Then the regression model would be the one shown in Equation \ref{eq 3.4} where $\operatorname{Var}(Y_{it}\vert X_{it},T_{it})=\lambda(1+\alpha\lambda)$; implying that $E(Y_{it}\vert X_{it},T_{it})=\lambda<\operatorname{Var}(Y_{it}\vert X_{it},T_{it})$.  Since we have that $\log r=X\beta^{'}+T\omega$ and that $\mathrm{v}:X,T\mapsto f(X,T)$, it follows then that $\exists$ a log link function $g$ that transforms the expectation of $Y$ to $X\beta^{'}+T\omega$.  Accordingly, formulating the model in this way is similar to specifying a generalized linear model with a log link function (\cite{r58,r39}).  The model can, therefore, also be written as shown in Equation (\ref{eq 3.5}).
\end{model}
\begin{model}\label{Model 2}
We let $r$ be composed of two data generating processes.  The first process is a Bernoulli process $a_{b}=(a_{1},a_{2},a_{3},\ldots)$ generating the structural zeros such that $\forall a_{i}\in a_{b}$ we have $a_{i}\in\{0,1\}$ where $a_{i}=1$ generates $Y_{it}=\{0\}$ when $\mathrm{v}$ is unobserved i.e. $\mathrm{v}=\emptyset$ and $a_{i}=0$ generates $Y_{it}=\{0\}$ when $\mathrm{v}=0$.  Furthermore, $a_{b}$ is such that $\forall a_{i}\in a_{b}$ we have that $Pr(a_{i}=1)=p$ and correspondingly $Pr(a_{i}=0)=1-p$.  Since $a_{i}=1$ generates $Y_{it}:\mathrm{v}=\emptyset$ and $a_{i}=0$ generates $Y_{it}:\mathrm{v}=0$, it follows then that the probability that $Y_{it}=0$ is constant i.e. $Pr(Y_{it}\vert \mathrm{v}=\emptyset)=p$ and $Pr(Y_{it}\vert \mathrm{v}=0)=1-p$.  

The second data generating process is a gamma--poisson process $z$ that generates $Y_{it}$ such that $z\sim\operatorname{poisson}(\theta)$ and $\theta\sim\operatorname{gamma} (z,p/1-p)$, and $\exists$ some $Y_{it}\in\mathbb{N}^{0}$ where $Y_{it}=0$.  In our case, $Y_{it}$ is such that $\exists Y_{\mathrm{v}=0},Y_{\mathrm{v}=\emptyset}\in Y$ for all $i$ and $t$ where $Y_{\mathrm{v}=0}=\{0\}$ and $Y_{\mathrm{v}=\emptyset}=\{0\}$ with $Y_{\mathrm{v}=0}$ representing day case charges and $Y_{\mathrm{v}=\emptyset}$ representing the observations for the individuals who were never hospitalized during the period of analysis.  We can, therefore, conceptualize $Y_{\mathrm{v}={0^{+}}}=\{0,1,2,\ldots,n\}\in\mathbb{N}^{*}\cup\{0\}\forall i$ and $t$ such that $Y_{\mathrm{v}=0}\subset Y_{\mathrm{v}={0^{+}}}$ and $Y_{\mathrm{v}=\emptyset}\not\subset Y_{\mathrm{v}={0^{+}}}$.  Given that $r=\{a_{b},z\}$, the corresponding empirical model $\log r=(.)$ is a two--part model.  The first part is a logit model that predicts $Pr(Y_{\mathrm{v}=\emptyset}\vert X,T)$ for all $i$ and $t$.  The second part is a Negative Binomial Model that predicts $Pr(Y_{\mathrm{v}={0^{+}}}\vert X,T)\forall i$ and $t$ by accounting for $Pr(Y_{\mathrm{v}=\emptyset}\vert X_{it},T_{it})$ and has its usual properties i.e. $\Gamma(s,\theta)$ and $\operatorname{Var}(Y_{\mathrm{v}={0^{+}}}\vert X,T)=\lambda(1+\alpha\lambda)\textgreater E(Y_{\mathrm{v}={0^{+}}}\vert X,T)$ for all $i$ and $t$.  The model can be summarized as Equations (\ref{eq 3.6}) and (\ref{eq 3.7}).
\end{model}
\begin{model}\label{Model 3}
We let $Y_{it}$ be truncated at zero such that $\exists Y_{0\not\in\mathrm{v}}\in Y$ for all $i$ and $t$ where $Y_{0\not\in\mathrm{v}}=\{1,2,\ldots,n\}\in\mathbb{N}$ and $Y_{0\not\in\mathrm{v}}\subset Y\forall i$ and $t$.  Since $Y_{0\not\in\mathrm{v}}\in\mathbb{Z}^{+}$ unlike $Y_{\mathrm{v}={0^{+}}}$ in Model \ref{Model 2} where $Y_{\mathrm{v}={0^{+}}}\in\mathbb{N}^{*}\cup\{0\}$, there exists $c$ that generates $Y_{0\not\in\mathrm{v}}$ such that $c\sim\operatorname{truncated--poisson}(\theta)$ and $\theta\sim\operatorname{gamma} (c,p/1-p)$.  Because $c$ effectively follows a truncated--Negative Binomial distribution, it can be shown that $E(Y_{0\not\in\mathrm{v}}\vert X,T)=\lambda[1-F_{NB}(0)]^{-1}$ for all $i$ and $t$, and $\operatorname{Var}(Y_{0\not\in\mathrm{v}}\vert X,T)$ is equal to $E(Y_{0\not\in\mathrm{v}}\vert X,T)/F_{NB}(0)^{\alpha}$ multiplied by the product of $1-[F_{NB}(0)]^{1+\alpha}$ and $E(Y_{0\not\in\mathrm{v}}\vert X,T)\forall i$ and $t$ where $F_{NB}(0)$ denotes the cdf. for the Negative Binomial distribution evaluated at $\mathrm{v}=0$ and $\alpha$ is the dispersion parameter indicating the extent of overdispersion in $Y_{0\not\in\mathrm{v}}$ (see, for example, \cite{r25}).  The model is expressed as Equation \ref{eq 3.8}.  
\end{model}
\begin{psm}\label{PSM}
We let $T\in\{0,1\}$ be such that $\exists T_{T=1}\subset T$ and $T_{T=0}\subset T$ so that $T_{T=1}$ represents the observations for $T$ when individual $i$ uses telecare at time $t$ and $T_{T=0}$ represents the observations for variable $T$ when individual $i$ does not use telecare at time $t$.  We let also $Y_{T=1}$ denote the outcomes for telecare users and $Y_{T=0}$ denote the outcomes for those individuals who do not use telecare $\forall i$ and $t$.  Given that $X=\{X_{1},X_{2},\ldots,X_{5}\}\in\mathbb{N}^{0}$ and $T\in\mathbb{Z}_{2}^{n}:\mathbb{Z}_{2}=\{0,1\}$, there exists $Pr(\widehat{T=1}\vert X)=\sigma_{1}$ and $Pr(\widehat{T=0}\vert X)=\sigma_{2}\forall i$ and $t$, also known as propensity scores, such that $0<\sigma_{1},\sigma_{2}<1$.  We then match each observation denoted by $T_{T=1}$ to a single observation denoted by $T_{T=0}$ whose propensity score is closest according to $X=\{X_{jit}\}_{j=1}^{5}$.  If $Y_{T=0},Y_{T=1}\perp T_{it}\vert X_{it}\Rightarrow Y_{T=0},Y_{T=1}\perp (\sigma_{1},\sigma_{2})$, then the average treatment effect of telecare is computed as $E(Y_{T=1}\vert X,T_{T=1})-E(Y_{T=0}\vert X,T_{T=0})\forall i$ and $t$.  For a substantive discussion of the technique, see \cite{r50}.
\end{psm}
\vspace{-6.5mm}
\begin{eqnarray}
g\left(Y_{it}\right)=\eta_{it}&=&\beta_0+\mathlarger{\mathlarger{\sum}}_{j=1}^{5}\beta_j X_{jit}+\omega T_{it}+\epsilon_{it}\label{eq 3.5}\\
Pr\left(Y_{\mathrm{v}=\emptyset}\vert X_{it},T_{it}\right)&=& \lambda\left(\eta\right):\eta=f\left(X_{it},T_{it}\right)\label{eq 3.6}\\
\log z &=&  \beta_{0}+\mathlarger{\mathlarger{\sum}}_{j=1}^{5}\beta_{j}X_{jit}+\omega T_{it}\label{eq 3.7}\\
\log c &=&  \beta_{0}+\mathlarger{\mathlarger{\sum}}_{j=1}^{5}\beta_{j}X_{jit}+\omega T_{it}\label{eq 3.8}
\end{eqnarray}
where $\Lambda$ is the logistic link function and $\epsilon$ is a stochastic random error term.

An important point to note is that Models (\ref{Model 1}), (\ref{Model 2}) and (\ref{Model 3}) make assumptions about the probability distributions of $r,c$ and $z$ but do not make any assumption about the distributions of $X$ and $T$.  The \hyperref[PSM]{PSM technique}, on the other hand, does not make any assumption about the probability distribution of the unobservable process that generates $\mathrm{v}$ but it matches the observations in the analytic sample into two groups with more or less similar distributions for $X$ and $T$.
\begin{pro}\label{Proposition 1}

Suppose that $Y=f(X,T)$ is the empirical specification of $V=v(.)$ in Theorem \ref{Theorem 1} and that \hyperref[subsection 2.2]{Assumptions 1--3} all hold.  Suppose also that $\exists Y_{T=1},Y_{T=0}\in Y\forall i$ and $t$ such that $Y_{T=0},Y_{T=1}\subset Y$.  Then we have that $Y\not\!\perp\!\!\!\perp T\vert X\forall i$ and $t$, and also that $E(Y_{T=0}\vert T=1)-E(Y_{T=0}\vert T=0)=\Upsilon$ with $\Upsilon\neq0$.

\end{pro}

Proposition \ref{Proposition 1} suggests that estimating the treatment effect of telecare using the \hyperref[PSM]{PSM technique} as well as the other experimental study designs--including quasi--experiments and randomized controlled trials--would not be appropriate.  In particular, it asserts that since the indirect utility function of a particular individual contains the individual's unobservable biological endowment as one of its arguments, we expect that telecare users differ from non-users as they (telecare users) typically have comparatively poor health hence their need to use telecare.  As such, we also expect their potential outcomes to be different; which implies that the decision whether or not to use telecare is not random.  This is consistent with the discussions in the previous section that the use of telecare is potentially endogenous and that we could also have a problem of unobserved heterogeneity to contend with.

Because we have only information on the length of stay in hospital for those individuals who were admitted to hospital or psychiatric care during the $2010/2011$ financial year, there could also be an additional problem of sample selection.  Furthermore, as is expected of linked datasets, sample selection could also be brought about by missing data.  We, therefore, have to estimate our econometric models with a strategy that controls for the potential endogeneity of telecare use, potential unobserved heterogeneity and potential selectivity bias as failure to do so would render our estimated treatment effects inconsistent.
\subsection{Model identification}\label{subsection 3.2}
We need to properly identify the variable for telecare use as well as the issue of sample selection in order to be able to estimate the treatment effect.  Since the matching exercise in the \hyperref[PSM]{PSM technique} may not always yield balanced samples, we impute the missing data with the average propensity score of similar subjects in the opposite group (see, for example, \cite{r28}).  For \hyperref[Model 1]{Models 1--3}, identification is done by way of instrumental variables.  The basic idea is that since telecare use is potentially endogenous because it is a choice variable and admission to hospital is not likely to be random, there are other variables besides our covariates of interest that are correlated with both the telecare variable and the likelihood of being admitted to hospital.  Since we are seeking to address both a potential endogeneity problem as well as a potential selectivity issue simultaneously, we need two instrumental variables for exact identification.  The variables that serve as instruments, however, should be strong predictors of telecare use and hospitalization but not determined in the three formulated models.  See \cite{r43} and \cite{r48} for a substantive discussion of model identification using instrumental variables. 

In this paper, we use the proportion of telecare users in each local council area in Scotland and the Scottish Index of Multiple Deprivation (SIMD) constructed as an ordinal variable with $10$ categories as instrumental variables.  The SIMD is a measure of area level deprivation for about $7,000$ small areas in Scotland (also referred to as data zones) in regard to employment, population, health, crime, access to services, housing and income distribution (see \cite{r70} for a more detailed discussion of the index).  The instrumental variable is constructed such that the lowest category, Category 1, is for the most well-off areas, whereas the highest category, Category 10, is for the most deprived areas.  We expect the proportion of telecare users to be related to telecare use in that the higher the proportion of telecare users in a particular local council area, the higher the likelihood that a particular individual who resides in that local council area used telecare, holding other factors constant.  Similarly, we expect an inverse relationship between the variable for telecare use and that for area level deprivation since the higher the level of deprivation, the lower the likelihood of accessing health and social care services including telecare.  We, however, do not expect our chosen instrumental variables to be determined in the main models since they are aggregated measures at the population level.
\subsection{Estimation strategy}\label{subsection 3.3}
We estimate all the econometric models using Maximum Likelihood Estimation (MLE) in STATA statistical software.  The \hyperref[PSM]{PSM technique} discussed in Section \ref{subsection 3.1} is implemented using the \texttt{teffects psmatch} command with the propensity scores in this case computed as the predicted probabilities of telecare use and not using telecare given the study covariates.  

In order to control for the potential endogeneity of telecare use and unobserved heterogeneity in the three Negative Binomial Models that we have formulated, we use the Two--Stage Residual Inclusion (2SRI) approach due to \cite{r57}.  This technique involves estimating the reduced form model of telecare use shown in Equation \ref{eq 3.9}--that relates our study covariates i.e. $X=\{X_{jit}\}_{j=1}^{5}$ and chosen instruments i.e. $Z=\{Z_{jit}\}_{j=1}^{2}$ to the variable for telecare use--using a probit model, obtaining the residuals of the model and then including, in the substantive models, the estimated residuals together with an interaction of the residuals with the telecare variable as controls for the potential endogeneity of telecare use and potential unobserved heterogeneity respectively.  The models are said to suffer from endogeneity and unobserved heterogeneity if their control terms are found to be statistically significant.  This approach has been shown to result in consistent estimates of the treatment effects in non-linear models unlike the other methods in the literature such as the Two--Stage Predictor Substitution (2SPS) technique and the IV estimator that entail replacing the treatment variable of interest with its fitted values (see, for example, \cite{r57}).

In order to control for potential selectivity bias, we first obtain the residuals\footnote{The residuals used to control for the potential endogeneity of telecare use and potential selectivity bias are also known as generalized or deviance residuals and are computed as $\phi(.)/\Phi(.)$ for the observations denoted by $T=1$ and $I=1$, and $-\phi(.)/1-\Phi(.)$ for the observations denoted by $T=0$ and $I=0$, where $\phi(.)$ and $\Phi(.)$ are the pdf. and cdf. for each observation respectively.  See \cite{r24} for a substantive discussion of generalized residuals.} of the sample selection probit model shown in Equation \ref{eq 3.10} that relates $X$ and $Z$ to an indicator of whether or not the outcome measure is observed, which is denoted by $I$.  We then include the residuals together with variable $I$ in our main models following \cite{r60}\footnote{The approach due to \cite{r60} that we use to correct for potential selectivity bias is very similar to the other approaches in the literature such as the approach due to \cite{r63} and the one due to \cite{r64} in the sense that all these techniques use elements of the selection equation to construct control function variables which serve as the correction terms for potential selectivity bias.  One may, therefore, use any of them as desired.}  Statistical significance of the coefficients of the control terms for potential selectivity bias is indicative of sample selection.

We also include a time trend variable $t$ in our substantive models to control for unexplained trend variations.  In Models (\ref{Model 1}), (\ref{Model 2}) and (\ref{Model 3}) this variable is constructed as the number of weeks before the March 2011 census because while the study data is generated as time series data, the Homecare Census dataset which contains the telecare variable covers only the March 2011 census week.  The coefficient of this variable, therefore, indicates the trend in the observed length of stay in hospital before the census week, holding the other covariates in the models constant.  An advantage of including the time trend variable from a theoretical standpoint is that even though individuals' biological endowments are unobservable, we may still be able to control for some variation due to these endowments since the time trend variable is observable.  The three Negative Binomial Models are thus specified as shown in Equations (\ref{eq 3.11}), (\ref{eq 3.12}) and (\ref{eq 3.13}).
\vspace{1mm}
\begin{eqnarray}
Pr\left(T_{it}=1\right) &=& \Phi\left(\beta_0+\overbrace{\mathlarger{\mathlarger{\sum}}_{j=1}^{5}\beta_j X_{jit}}^\text{study covariates}+\mathlarger{\mathlarger{\sum}}_{j=1}^{2}\omega_{j}Z_{jit}+\omega_{t}t_{it}\right)\label{eq 3.9}\\
Pr\left(I_{it}=1\right) &=& \Phi\left(\beta_0+\mathlarger{\mathlarger{\sum}}_{j=1}^{5}\beta_j X_{jit}+\overbrace{\mathlarger{\mathlarger{\sum}}_{j=1}^{2}\omega_{j}Z_{jit}}^\text{instruments}+\omega_{t}t_{it}\right)\label{eq 3.10}
\end{eqnarray}

\begin{equation}
\log r =\beta_{0}+\mathlarger{\mathlarger{\sum}}_{j=1}^{5}\beta_{j}X_{jit}+\omega T_{it}+\underbrace{\omega_{{t_{1}}}t_{{1_{it}}}+\omega_{\xi}\xi_{it}+\omega_{T\xi}T\xi_{it}+\omega_{{\xi_{s}}}\xi_{{s_{it}}}+\omega_{I}I}_\text{control function terms for potential bias}\label{eq 3.11}
\end{equation}
\begin{equation}
\log z =\beta_{0}+\mathlarger{\mathlarger{\sum}}_{j=1}^{5}\beta_{j}X_{jit}+\omega T_{it}+\underbrace{\omega_{{t_{1}}}t_{{1_{it}}}+\omega_{\xi}\xi_{it}+\omega_{T\xi}T\xi_{it}+\omega_{{\xi_{s}}}\xi_{{s_{it}}}+\omega_{I}I}_\text{control function terms for potential bias}\label {eq 3.12}
\end{equation}
\begin{equation}
\log c =\beta_{0}+\mathlarger{\mathlarger{\sum}}_{j=1}^{5}\beta_{j}X_{jit}+\omega T_{it}+\underbrace{\omega_{{t_{1}}}t_{{1_{it}}}+\omega_{\xi}\xi_{it}+\omega_{T\xi}T\xi_{it}+\omega_{{\xi_{s}}}\xi_{{s_{it}}}+\omega_{I}I}_\text{control function terms for potential bias}\label {eq 3.13}
\end{equation}
\vspace{0.5mm}
where Equations (\ref{eq 3.11}), (\ref{eq 3.12}) and (\ref{eq 3.13}) are our three Negative Binomial Models of interest discussed in Section \ref{subsection 3.1}.  In particular, Equation \ref{eq 3.11} i.e. $\log r=(.)$ is the basic Negative Binomial Model represented by Model \ref{Model 1}, Equation \ref{eq 3.12} i.e. $\log z=(.)$ is the zero--inflated Negative Binomial Model represented by Model \ref{Model 2}, and Equation \ref{eq 3.13} i.e. $\log c=(.)$ is the zero--truncated Negative Binomial Model represented by Model \ref{Model 3}.
\begin{pro}\label{Proposition 2}

Suppose that (i) Assumption \ref{Assumption 1}, Assumption \ref{Assumption 2} and Assumption \ref{Assumption 3} hold; (ii) the indirect utility function $V=v(.)$ can be estimated as $Y=f(X,T,t)$; (iii) the demand function for telecare, $T=t(.)$, can be estimated as $T=k(X,Z,t)$ with $\xi$ as the estimated residuals of the demand equation, and (iv) $I\in\{0,1\}=g(X,Z,t)$ is a selection equation for $Y:I=1$ if $Y\in\operatorname{supp}(Y)$ and $I=0$ otherwise with $\xi_{s}$ as the estimated residuals of the selection equation.  Suppose also that (i) $\exists G_{TX}=\{X_{m\times n},T_{m\times 1}\}\in\mathbb{R}_{m\times n}$ with $\operatorname{rk}(n)$ that explains almost all the variation in $Y$ such that $Y-\hat{Y}\approx0$ and (ii) $Y$ can be estimated using strategies $\mathrm{s}_{1}=\{Y\vert X,T\}$, $\mathrm{s}_{2}=\{Y\vert X,\hat{T},t\}$, $\mathrm{s}_{3}=\{Y\vert X,T,t,\xi\}$, $\mathrm{s}_{4}=\{Y\vert X,T,t,\xi,T\xi\}$ and $\mathrm{s}_{5}=\{Y\vert X,T,t,\xi,T\xi,\xi_{s},I\}$ which are defined as follows:
\begin{eqnarray*}
\mathrm{s}_{1}:Y&=& \omega+X{\scriptstyle_{m\times n}}.\mathrm{b}{\scriptstyle_{X}}+T{\scriptstyle_{m\times 1}}.\mathrm{b}{\scriptstyle_{T}}+\epsilon_{1}\\
\mathrm{s}_{2}:Y&=&\omega+X{\scriptstyle_{m\times n}}.\mathrm{b}{\scriptstyle_{X}}+\hat{T}{\scriptstyle_{m\times 1}}.\mathrm{b}{\scriptstyle_{\hat{T}}}+t.\mathrm{b}{\scriptstyle_{t}}+\epsilon_{2}\\
\mathrm{s}_{3}:Y&=&\omega+X{\scriptstyle_{m\times n}}.\mathrm{b}{\scriptstyle_{X}}+T{\scriptstyle_{m\times 1}}.\mathrm{b}{\scriptstyle_{T}}+t.\mathrm{b}{\scriptstyle_{t}}+\xi.\mathrm{b}{\scriptstyle_{\xi}}+\epsilon_{3}\\
\mathrm{s}_{4}:Y&=&\omega+X{\scriptstyle_{m\times n}}.\mathrm{b}{\scriptstyle_{X}}+T{\scriptstyle_{m\times 1}}.\mathrm{b}{\scriptstyle_{T}}+t.\mathrm{b}{\scriptstyle_{t}}+\xi.\mathrm{b}{\scriptstyle_{\xi}}+T\xi.\mathrm{b}{\scriptstyle_{T\xi}}+\epsilon_{4}\\
\mathrm{s}_{5}:Y&=&\omega+X{\scriptstyle_{m\times n}}.\mathrm{b}{\scriptstyle_{X}}+T{\scriptstyle_{m\times 1}}.\mathrm{b}{\scriptstyle_{T}}+t.\mathrm{b}{\scriptstyle_{t}}+\xi.\mathrm{b}{\scriptstyle_{\xi}}+T\xi.\mathrm{b}{\scriptstyle_{T\xi}}+\xi_{s}.\mathrm{b}{\scriptstyle_{{\xi_{s}}}}+I.\mathrm{b}_{I}+\epsilon_{5}
\end{eqnarray*}
Then (i) $V_{t-1}\equiv V_{t}$ when estimating $Y$ using $\mathrm{s}_{1}$; (ii) $\mathrm{s}_{1}$ and $\mathrm{s}_{2}$ are not sufficient; (iii) $E(I,\xi_{s})\approx E(T,\xi)$ when estimating $Y$ using $\mathrm{s}_{3}$, and (iv) $E(I,\xi_{s})\approx0$ when estimating $Y$ using $\mathrm{s}_{4}$.

\end{pro}

Proposition \ref{Proposition 2} shows that our estimation strategy which entails controlling for the potential endogeneity of telecare use, potential unobserved heterogeneity, potential selectivity bias and unexplained trend variations in tandem results in a better estimate of the treatment effect than that which is obtained from an experimental study design in the sense that it controls for the influence of unobservable factors in a general way.  It states that even if it were possible for us to estimate the expected length of stay in hospital using the best possible predictors of the outcome measure, we would still not be able to account for all the confounding as there are unobservable factors that influence the outcome.  Furthermore, if we were to do our estimation at only one point in time, then we would inadvertently assume that individuals' utility functions are time invariant.  Yet this is not likely to be the case as we expect individuals' stock of health capital to depreciate over time, holding other factors constant.  

In the event that we were to purge all unobservable confounding that causes systematic variation in a particular individual's expected length of stay in hospital, by for instance substituting $T$ with its fitted values, $\hat{T}$, then this would not be plausible as it would be akin to having identical biological endowments for telecare users.  If we were to assume that individuals are admitted to hospital just by random chance and it turned out to be the wrong assumption to make, then any attempt to estimate the treatment effect would not be sufficient even when using the best possible predictors of the outcome and accounting for the possibility of systematic variations in the outcome measure due to individuals' unobservable biological endowments.
\section{Empirical illustration}\label{Section 4}

\subsection{Description of the study covariates}\label{subsection 4.1}
In this section, we describe the study covariates for the three Negative Binomial Models formulated in Section \ref{Section 3} and the \hyperref[PSM]{PSM technique}.  Table \ref{Table 1} on page 22 gives the variable definitions while Table \ref{Table 2} on page 23 presents the descriptive statistics.  The data for the study in this paper is such that there are certain individuals who were not admitted to hospital during the $2010/2011$ financial year and we are, therefore, not able to observe their outcomes; there are other individuals who were admitted to hospital during the period of analysis but did not use homecare services and we are thus not able to observe their treatment status, and there are some other individuals who were admitted to hospital but we do not observe some of their covariates due to missingness.  Since the outcome measure is available only in the SMRs, the estimation samples for our analyses can be drawn only for the individuals who were admitted to hospital.

According to Table \ref{Table 1}, we note that there are several types of variables in that there are some variables with categories, whereas there are other variables that do not have categories.  There are also some variables that represent counts, whereas there are other variables that are continuous.  Of the categorical variables, we also note that there are some indicator variables--also known as binary variables--and other variables with multiple categories.  In particular, the table shows that there are five binary variables i.e. area of residence, telecare use, sex, comorbidity status and inclusion into the sample; two continuous variables i.e. age and proportion of telecare users in each local council area; three count variables i.e. the two time trend variables and the variable for the length of stay in hospital; one nominal variable i.e. client group, one ordinal variable i.e SIMD-decile, and a variable that identifies each individual in the dataset i.e. project ID.

Table \ref{Table 2} shows the distributions of our variables of interest in the entire population of homecare clients i.e. the study population and in the population of individuals who were admitted to hospital during the analytic period i.e. the study sample.  A closer look at Table \ref{Table 2} shows that the table has five columns.  The first column contains the variable names; the second column contains the number of observations for each variable in the study population; the third column contains the proportions of the indicator variables--or the mean values for the continuous variables and the median values for the count and ordinal variables--in the study population; the fourth column contains the number of observations for each variable in the study sample, while the fifth column contains the proportions of the indicator variables--or the mean values for the continuous variables and the median values for the count and ordinal variables--in the study sample.
\begin{table*}
\caption{Variable Definitions}
\label{Table 1}
\begin{tabular}{ll}
\hline
 \\
Variable & \multicolumn{1}{l}{Definition}\\
\hline
Age & Age in years at time $t$. \\
Sex & $1$ if female, $0$ otherwise. \\
Area of residence & 1 if a particular individual was living in a rural \\
 & area at time $t$, 0 otherwise. \\
Client group &    1 if a particular individual had a diagnosis\\
& of dementia or other mental illness,\\
 & 2 if a particular individual had a learning disability,\\
 & 3 if a particular individual had a physical disability,\\
 & 4 if a particular individual was frail.\\
Telecare use & 1 if a particular individual used telecare devices, \\
 & 0 otherwise.\\
Comorbidity status & 1 if a particular individual had three or more\\
 & comorbid conditions at time $t$, 0 otherwise.\\
Length of stay in hospital & The number of days that a particular  \\
 & individual stays in hospital to receive treatment.\\
Inclusion into the sample & 1 if the outcome variable for a particular individual\\
 & at time $t$ is included in the estimation sample,\\
& 0 otherwise.\\
SIMD-decile & 10 categories of the Scottish Index of Multiple\\
 & Deprivation in ascending order. \\
Proportion of telecare users  & A variable indicating the proportion of telecare \\
 & users in each local council area.\\
Project ID & A unique reference number for each individual in \\
 & the dataset.\\
Telecare residuals & The residuals obtained from a reduced form model\\
 & of telecare use.\\
Control for sample selection bias & A control for potential sample selection bias \\
 & following \cite{r60}.\\
Time trend & A time trend variable where the unit of time is\\
& 1 week.\\
Time trend 1 & A count variable for the number of weeks before the \\
 & March 2011 census.\\
\\ \hline
\end{tabular}
\end{table*}

The results in Columns (3) and (5) of Table \ref{Table 2} show that the mean age of the homecare clients in both the study population and study sample is approximately $75$ years.  The descriptive statistics also show that about half of the total observations belonged to the individuals who were considered to be frail; approximately $16\%$ and $17\%$ of the observations in the study population and study sample respectively were for the physically disabled homecare clients; about $6\%$ of the observations in both the study population and study sample were for the individuals who had a diagnosis of dementia or other mental illness, and approximately $30\%$ and $27\%$ of the observations in the study population and study sample respectively were for the individuals who had learning disabilities.  The results further show that approximately $40\%$ of the total observations belonged to the individuals with three or more comorbid conditions; implying that a substantial proportion of the population were living with multiple conditions.
\begin{table*}
\begin{threeparttable}
\caption{Descriptive statistics}
\label{Table 2}
\begin{tabular}{lcccc}
\hline
\\
& \multicolumn{2}{c}{Study population} & \multicolumn{2}{c}{Study sample} \\
\\\cmidrule(lr){2-3}\cmidrule(lr){4-5}
 & \multicolumn{1}{c}{$n$} & \multicolumn{1}{c}{$m$} &  \multicolumn{1}{c}{$n$} & \multicolumn{1}{c}{$m$}\\
\hline
Age & 48,749 & 75.5 & 21,580 & 75.0 \\
Male  & 48,749 & 0.35 & 21,580 & 0.35 \\
Female  & 48,749 & 0.65 & 21,580 & 0.65 \\
Area of residence &    48,584 & 0.09 & 21,517 & 0.09 \\
Telecare use & 49,025 & 0.11 & 21,580 & 0.10 \\
Dementia and Mental Health & 49,025 & 0.06 & 21,580 & 0.06 \\
Learning disability & 49,025 & 0.30 & 21,580 & 0.27 \\
Physical disability & 49,025 & 0.16 & 21,580 & 0.17 \\
Frail elderly & 49,025 & 0.48 & 21,580 & 0.50 \\
Comorbidity status & 49,025 & 0.36 & 21,580 & 0.43
\\
SIMD-decile & 48,644 & 5.00 & 21,553 & 5.00  \\
Length of stay in hospital & 26,285 & 4.00 & 21,580 & 2.00 \\
Time trend & 49,025 & 40.0 & 21,580 & 28.0  \\
Time trend 1 & 49,025 & 12.0 & 21,580 & 24.0  
\\ \hline
Number of homecare clients & 25,982 &  & 10,590 &   
\\ \hline
\end{tabular}
 \begin{tablenotes}

            \item 
\begin{FlushLeft}
Notes: $n$ = number of observations; $m$ = arithmetic mean or median or proportion where applicable.
\end{FlushLeft}
\end{tablenotes}
\end{threeparttable}
\end{table*}
\subsection{Empirical results}\label{subsection 4.2}
In this section, we present the empirical results of our econometric models of interest.  Table \ref{Table 3} on page 24 contains the estimated coefficients of the reduced form model of telecare use shown in Equation \ref{eq 3.9} and the sample selection model shown in Equation \ref{eq 3.10}.  Table \ref{Table 4} on page 26 contains the estimated average treatment effect of telecare on the length of stay in hospital using the \hyperref[PSM]{PSM technique} as well as the exponentiated coefficients of our three models of interest.  Table \ref{Table 5} on page 27 contains the empirical results of five variants of Model \ref{Model 1}.
\vspace{-50mm}
\begin{table*}[!ht]
\begin{threeparttable}
\caption{The first stage models}
\label{Table 3}
\begin{tabular}{lcc}
\hline
\\
& \multicolumn{1}{c}{$(1)$} & \multicolumn{1}{c}{$(2)$} \\
\\\cmidrule(lr){2-2}\cmidrule(lr){3-3}
 & \multicolumn{1}{c}{$T=1$} & \multicolumn{1}{c}{$I=1$} \\
\hline
Age & 0.006 & 0.002  \\
    & (5.24) & (1.58)\\
Square of age & $-$0.0001 & $-$0.00003  \\
    & (5.58)  & (3.36)\\
Sex &    $-$0.008 & 0.009 \\
    & (1.24)  & (1.54)\\
Area of residence & 0.011 & $-$0.009  \\
    & (0.97) & (0.89)\\
Dementia and Mental Health & 0.042 & $-$0.183 \\
    & (5.81) & (25.76)\\
Learning disability & 0.032 & $-$0.004  \\
    & (2.78) & (0.38)\\
Physical disability & 0.001 & $-$0.043 \\
   & (0.10) & (5.01)\\
Comorbidity status & $-$0.005 & $-$0.066  \\
    & (0.83)& (11.99)\\
SIMD-decile & $-$2.033 & $-$0.636 \\
    & (9.77)& (2.89)\\
Proportion of telecare users & 1.219 & $-$0.681\\
    & (27.00) & (17.98)\\    
Time trend & 0.0001 & $-$0.008 \\
    & (0.72) & (51.25) 
\\ \hline
Number of observations & 48,571 &  48,571
\\ \hline
Number of homecare clients & 25,598 &  25,598
\\ \hline
\end{tabular}
 \begin{tablenotes}

            \item 
Notes: The table presents the average marginal effects for the reduced form model of telecare use and the sample selection model with the Robust $Z$ statistics in parenthesis.  The standard errors used to compute the $Z$ statistics are clustered by Project ID.  $Z$ statistics greater than or equal to $1.96$ imply statistical significance at $5\%$ level of significance.  We include the square of age as an additional explanatory variable to control for the potential non-linear effect of age on the outcome variables. The reference category for client group is `Frail elderly'.
\end{tablenotes}
\end{threeparttable}
\end{table*}
\vspace{50mm}

The estimated coefficients of the reduced form model of telecare use and the sample selection model are presented as average marginal effects.  The average marginal effect of a particular covariate is interpreted as the change in the probability of observing the outcome measure due to a unit change in the covariate (\cite{r37}).  According to the results in the table, we can observe that the two instrumental variables have their expected signs and are statistically significant at $5\%$ level of significance.  In particular, the results show that, holding other factors constant, the variable for the proportion of telecare users in each local council area is directly related to telecare use, whereas the higher the area level deprivation, the lower the likelihood of using telecare controlling for the other covariates in the model.  The results also show that the variables for age and client group are statistically significant at $5\%$ significance level.  More specifically, the results indicate that an increase in age by one year is associated with an increase in the probability of telecare use by 0.006, other factors held constant.  We can also observe from the table that the individuals with dementia or other mental illnesses and those with learning or physical disabilities are more likely to use telecare than their counterparts, although the effect of `physical disability' on telecare use is not statistically significant at $5\%$ level of significance.  Furthermore, the results also show that about $0.01\%$ of the observed probability of telecare use is unexplained by the model, other factors held constant, since the average marginal effect of the time trend variable is $0.0001$ and statistically significant at $5\%$ significance level.

Looking at Column 2 of the table, we note that the probability of inclusion into the sample for individuals with dementia or other mental illnesses is approximately $18\%$ lower than that for the frail elderly, holding other factors constant.  We also note that the average marginal effect of comorbidity status is $-0.066$; implying that a particular individual with three or more comorbid conditions has a lower probability of being included in the sample than another individual with fewer comorbid conditions by about $0.7\%$, holding other factors constant.  The results further show that there is a general decrease in the probability of sample selection over time by $0.8\%$ since the average marginal effect of the time trend variable is $-0.008$ and statistically significant at $5\%$ level of significance; the higher the area level deprivation, the lower the probability of sample selection, other factors held constant, and the higher the proportion of telecare users in a particular local council area, the lower the likelihood that an individual who was residing in that local council area was included in the sample, all things equal. 
\begin{table*}[!ht]
\begin{threeparttable}
   \caption{Empirical results for Models 1-3 and the PSM technique}
   \label{Table 4}
   \begin{tabular}{l c c c }
   \hline
\\
Variable & \multicolumn{1}{c}{$(1)$} & \multicolumn{1}{c}{$(2)$} & \multicolumn{1}{c}{$(3)$} \\
\hline
   Age & $1.042$ & $1.048$ & $1.054$\\
       & ($5.76$) & ($6.01$) & ($4.93$)\\
   Square of age & $0.999$ & $0.999$ & $0.999$\\
       & ($7.33$) & ($7.73$) & ($6.39$)\\     
   Comorbidity status & $2.639$ & $2.437$ & $2.090$\\
       & ($10.48$) & ($10.45$) & ($6.65$)\\    
   Dementia and Mental Health & $0.510$ & $0.483$ & $0.420$\\
       & ($11.17$) & ($11.65$) & ($10.63$)\\   
   Learning disability & $1.699$ & $1.484$ & $2.091$\\
       & ($6.88$) & ($4.63$) & ($6.65$)\\   
   Physical disability & $0.964$ & $0.933$ & $0.970$\\
       & ($0.78$) & ($1.37$) & ($0.45$)\\   
   Telecare use & $0.568$ & $0.620$ & $0.488$\\
       & ($4.27$) & ($3.18$) & ($3.87$)\\    
   Sex & $0.887$ & $0.915$ & $0.863$\\
       & ($3.64$) & ($2.53$) & ($3.09$)\\    
   Area of residence & $1.245$ & $1.244$ & $1.375$\\
       & ($4.31$) & ($4.04$) & ($4.29$)\\    
   Time trend 1 & $1.035$ & $1.036$ & $1.038$\\
       & ($8.09$) & ($8.03$) & ($7.26$)\\    
   Telecare residuals & $0.013$ & $0.011$ & $0.04$\\
       & ($10.45$) & ($10.22$) & ($10.22$)\\    
   Selection equation residuals & $0.003$ & $0.002$ & $0.001$\\
       & ($11.13$) & ($11.16$) & ($11.15$)\\    
   Telecare interacted with residuals & $38.362$ & $19.560$ & $88.104$\\
       & ($3.76$) & ($2.71$) & ($3.41$)
\\ \hline   
\textbf{Wald test for weak} & 824.41& & \\ 
\textbf{instruments}; $\chi^2(p-value)$ &(0.00) & & \\\hline
\textbf{Wald Chi-square test;} & $411.78$ & $391.38$ & $47.03$ \\
 ${\chi}^2$, (p-value) &($0.0000$) &($0.0000$)& ($0.0000$) \\\hline
\textbf{Likelihood ratio test} & $\alpha=1.142$ & $\alpha=1.512$ & $\log\alpha=22.910$ \\
    for $\alpha=0$ or $\log\alpha=-\infty(p-value)$ &($0.00$) &($0.00$) & ($0.00$) \\\hline  
\textbf{PSM technique}& $-0.439$ & & \\
 $\operatorname{ATE}(p-value)$ & (0.508) & & \\ \hline
\textbf{Number of observations} & $15,157$ & $21,511$ & $15,157$ \\ \hline
\textbf{Number of homecare clients} & $8,590$ & $10,586$ & $8,590$ \\ \hline

   \end{tabular}
   \begin{tablenotes}
      \item 
Notes: ATE=Average Treatment Effect.  The variable for inclusion into the sample was omitted due to collinearity.  Robust $Z$-statistics are in parentheses and the standard errors used to compute them have been clustered by Project ID.  $Z$-statistics $\geq1.96$ imply that the Incidence Rate Ratio is statistically significant at $5\%$ significance level.  The reference category for client group is `Frail elderly'. 
   
    \end{tablenotes}
\end{threeparttable}
\end{table*}

The exponentiated coefficients of our formulated Negative Binomial Models in Table \ref{Table 4} are presented as Incidence Rate Ratios (IRRs).  The IRRs are computed as $\exp(\beta)=d$, where $d$ is the IRR and $\beta$ is the estimated coefficient.  $d$ is interpreted as a change in the outcome measure by $(\exp(\beta)-1)\%$ for every unit change in a particular covariate (see, for example, \cite{r61}).  The second column of the table, labeled (1) contains the IRRs for Model \ref{Model 1}.  The third column, labeled (2) contains the IRRs for Model \ref{Model 2}.  The fourth column, labeled (3) contains the IRRs for Model \ref{Model 3}.  The average treatment effect derived using the \hyperref[PSM]{PSM technique} is indicated below the empirical results of Model \ref{Model 1}.

In order to validate the three formulated models, a preliminary analysis was conducted to assess the appropriateness of generating time series data for our empirical illustration by investigating parameter stability over time.  In particular, the analysis entailed interacting the covariates in the econometric models of primary interest with the variable for time trend, including the interaction terms as additional regressors in the model, and then inspecting the coefficients of the interaction terms (see, for example, \cite{r49}).  The results of this exercise showed that the IRRs of the interaction terms are either not statistically significant at $5\%$ level of significance or are approximately equal to one, which implies that the effects of our study covariates have remained more or less the same over the $2010/2011$ financial year.  See the supplementary material for a more comprehensive discussion of the assessment and a full specification of the model.

A joint significance test was also conducted on the estimated coefficients of the two instrumental variables in the reduced form model of telecare use and the sample selection model to test for weak instruments.  The results of this test (presented in Table \ref{Table 4}) show that the instrumental variables are not weak since their estimated coefficients are significantly different from zero ($\chi^{2}=824.4;\operatorname{p-value}=0.0000$).  We also conducted a likelihood ratio test on the dispersion parameter, $\alpha$, to determine whether or not the Negative Binomial Model is the appropriate one to use in our case (\cite{r38}).  The results of this test (presented in Table \ref{Table 4}) confirm the appropriateness of using the Negative Binomial Model since the dispersion parameter is greater than one and statistically significant across all the formulated models. 
\begin{table*}[!ht]
\begin{threeparttable}
\caption{Incidence Rate Ratios for five variants of Model 1}
\label{Table 5}
\begin{tabular}{lccccc}
\hline
\\
Variable & \multicolumn{1}{c}{$(1)$} & \multicolumn{1}{c}{$(2)$} & \multicolumn{1}{c}{$(3)$} & \multicolumn{1}{c}{$(4)$} &
\multicolumn{1}{c}{$(5)$} \\
\hline
Age & 1.001 & 1.001 & 1.003 & 1.042 & 1.042 \\
    & (0.08) & (0.09) & (0.29) & (5.87) & (5.76)\\
Square of age & 0.999 & 0.999 & 0.999 & 0.999 & 0.999 \\    
    & (0.32) & (0.40) & (0.64) & (7.38) & (7.33)\\
Comorbidity status &    1.010 & 1.002 & 1.003 & 2.330 & 2.639 \\
    & (0.27) & (0.07) & (0.09) & (10.25) & (10.48)\\
Dementia and Mental Health & 0.890 & 0.894 & 0.890 & 0.516 & 0.510 \\   
    & (2.16) & (2.03) & (2.11) & (10.96) & (11.17)\\ 
Learning disability & 1.684 & 1.648 & 1.676 & 1.734 & 1.699 \\  
    & (5.49) & (5.19) & (5.33) & (7.00) & (6.88)\\ 
Physical disability & 1.061 & 1.062 & 1.059 & 0.962 & 0.964 \\ 
    & (1.09) & (1.10) & (1.04) & (0.82) & (0.78)\\
Telecare use & 0.999 & 1.008 & 1.022 & 0.937 & 0.568 \\        
    & (0.01) & (0.09) & (0.24) & (1.19) & (4.27)\\
Sex & 1.013 & 1.020 & 1.015 & 0.889 & 0.877 \\        
    & (0.34) & (0.52) & (0.37) & (3.55) & (3.64)\\  
Area of residence & 1.119 & 1.127 & 1.119 & 1.244 & 1.245 \\ 
    & (2.22) & (2.33) & (2.19) & (4.17) & (4.31)\\         
Time trend 1 & & 0.991 & 0.991 & 1.034 & 1.035 \\ 
   &  & (5.87) & (5.89) & (7.81) & (8.09)\\    
Telecare residuals & & & 0.644 & 0.018 & 0.013 \\ 
   &  &  & (1.38) & (9.80) & (10.45)\\
Selection equation residuals & & & & 0.003 & 0.003 \\ 
   &  &  &  & (10.84) & (11.13)\\  
Telecare interacted with residuals & & & & & 38.362 \\
   &  &  &  &  & (3.76)             
\\ \hline
Number of observations & 15,158 & 15,158 & 15,157 & 15,157 & 15,157 \\ 
\hline
Number of homecare clients & 8,591 & 8,591 & 8,590 & 8,590 & 8,590 
\\ \hline
\end{tabular}    
\begin{tablenotes}
      \item 
Notes: The variable for inclusion into the sample was omitted due to collinearity.  Robust $Z$-statistics are in parentheses and the standard errors used to compute them have been clustered by Project ID.  $Z$-statistics greater than or equal to $1.96$ imply that the Incidence Rate Ratio is statistically significant at $5\%$ significance level.  The reference category for client group is `Frail elderly'.     
   
    \end{tablenotes}

\end{threeparttable}
\end{table*}
A further joint significance test on the estimated coefficients of our three Negative Binomial Models of interest was conducted using the Wald Chi-square test.  The results of this test (presented in Table \ref{Table 4}) show that the explanatory variables included in the formulated models help in explaining the observed variation in the length of stay in hospital

The empirical results of the \hyperref[PSM]{PSM technique} presented in Table \ref{Table 4} show that the use of telecare reduces the expected length of stay in hospital but it is not a significant predictor.  In particular, the average treatment effect is $-0.439$ with a $\operatorname{p-value}$ of $0.508$.  Looking at the other empirical results presented in the table, we can observe that the IRRs are robust across the three models in the sense that the effects of the covariates remain unchanged.  In particular, we note that older individuals, males, individuals residing in rural areas, individuals with learning disabilities and those with three or more comorbid conditions have a comparatively long expected length of stay in hospital, holding other factors constant, whereas individuals with dementia or other mental illnesses and telecare users have a relatively shorter expected length of stay in hospital than their counterparts, all else equal.  The results also show that it is important to control for the endogeneity of telecare use, unobserved heterogeneity and selectivity bias in tandem as the coefficients of all the control function variables are statistically significant at $5\%$ level of significance.

The results in Table \ref{Table 5} confirm the importance of controlling for all the estimation issues simultaneously.  According to the table, we note that the variable for telecare use is significant only in the variant of Model \ref{Model 1}, labeled (5), that controls for unobserved heterogeneity.  Furthermore, the results suggest that if we were to use the other estimation strategies besides the one that is used in the column labeled (5), then the IRR of telecare would be approximately equal to one; implying that telecare use would have been found to have very little effect on the expected length of stay in hospital.
\section{Conclusions}\label{Section 5}
This paper presents three econometric models that can be used to estimate the effect of telecare use on the length of stay in hospital.  A theory of the demand for telecare is first developed to explain a particular individual's relative preference for telecare in light of the prices faced by the individual and the individual's wealth level.  Since our analytical problem is such that there are some goods and services which are obtained from the market place and others for which no market exists, we make several necessary assumptions that allow us to properly identify individuals' demand behavior.  We conceptualize telecare as a health input in a particular individual's health production function that also enables the individual to maximize utility.  Accordingly, we have that telecare users choose to use telecare instead of the other competing services in order to achieve welfare gains via health production.

Because the true health status of a particular individual is unknown to us largely due to the individual's unobservable biological component, the health production function is defined as a function of both observable health inputs such as telecare as well as the individual's uonbservable biological component.  We prove that the demand for telecare can be decomposed into six components that explain the relationship between the quantities of the utility generating inputs consumed and their price levels in different scenarios.  We show that for a particular individual who uses telecare, the consumption of telecare increases when the economic conditions improve due to the individual's relative preference for telecare.  Moreover, the theory developed in this paper posits that the individuals who choose to use telecare would demand more telecare if the resources used to acquire the other utility generating inputs were to be reallocated to acquiring telecare.  Since telecare is one of the individual's observable health inputs, the increased consumption of telecare would imply that the individual has become more efficient at producing health, ceteris paribus.  The demand function for telecare also acknowledges that we are not able to observe all the factors that determine a particular individual's choice to use telecare due to our inability to observe the individual's unobservable endowment.

We also prove that the consumption of telecare for a particular individual not only depends on its own price and the individual's level of wealth, but also on the prices of the other utility generating inputs.  Given the specification of the indirect utility function in Theorem \ref{Theorem 1}, perhaps the most important result is Corollary \ref{Corollary 1}.  Here, we show that because part of a particular individual's utility function is unobservable, we cannot sufficiently explain the individual's utility maximizing behavior even with perfect data.  The implication of this result is that the study design and estimation strategy that one uses to estimate the impact of telecare matters since the unobservable component of the indirect utility function is handled differently.  In particular, if one uses an experimental study design to elicit the causal effect, then the design would not be able to randomize the unobservable component.  In the event that there are systematic variations in each study arm, then the estimated treatment effect would further be compromised.  We show in Proposition \ref{Proposition 1} that the expected differences between telecare users and non-users render experimental study designs insufficient as some of these differences are typically unobservable.  We, therefore, conjecture that this could be the reason why the randomized controlled trials in the literature that look into the effectiveness of telecare do not find significantly different health and service outcomes between telecare users and non-users (see, for example, \cite{r29,r27,r54}).

In this paper, we employ an estimation strategy that allows us to control for the unobservable component in the health production function in a way that an experimental study design cannot.  We use a flexible control function approach that controls for unobservable systematic variations within the population of telecare users in addition to incidental truncation using several econometric techniques.  This approach has been used by several other studies in the literature to address potential endogeneity, potential unobserved heterogeneity and potential selectivity bias in their models (see, for example, \cite{r56,r42,r47,r44,r11}).  We however extend the literature, particularly those studies that estimate the health production function whether directly or indirectly, by formulating a model that considers an individual's health production to be a dynamic process.  Since we expect individuals' stock of health capital to decrease over time, all else equal, we prove in Proposition \ref{Proposition 2} that the estimation strategies that do not employ time series analyses make a theoretically invalid assumption by assuming that individuals' utilities and, by extension, their stock of health capital remain fixed over time.

We also show in Proposition \ref{Proposition 2} that the strategies that entail substituting the treatment variable with an estimate of its predicted values are not sufficient in estimating the treatment effect as they do not control for all the relevant unobservable factors.  \cite{r57} also made a similar observation as they noted that the 2SRI approach is consistent in non-linear models, whereas the 2SPS technique is not.  We, therefore, conjecture that using the fixed effects regression model to account for bias in panel data may also not be ideal as it accounts only for the influence of individual level time-invariant unobservable factors.  Since the overall objective in this paper is to relate the demand for telecare to the length of stay in hospital, we make various assumptions about the probability distribution of the outcome measure and formulate three variants of the Negative Binomial Model i.e. a basic Negative Binomial Model, a zero--inflated Negative Binomial Model and a zero--truncated Negative Binomial Model.  We then use linked administrative health and social care data to estimate their treatment effects and compare the treatment effects with that of the \hyperref[PSM]{PSM technique}, which adopts a quasi-experimental study design.

A potential limitation of the study data, however, is that we are not able to observe all the covariate levels for those who were not admitted to hospital due to missingness and, as such, the effective estimation samples do not include their observations.  Accordingly, the basic Negative Binomial Model does not consider zero counts, whereas the zero--inflated component of the zero--inflated Negative Binomial Model does not contain the outcomes of the individuals who were never hospitalized during the analytic period.  Future research could, therefore, employ better data.  Nevertheless, we have no reason to believe that our estimated results would be different with better data as the missing data neither precludes the estimation of consistent treatment effects nor invalidates the results of the basic Negative Binomial Model and the zero--inflated model variant.  The empirical results of all the three Negative Binomial Models and the \hyperref[PSM]{PSM technique} show that telecare use is expected to reduce the length of stay in hospital, holding other factors constant, but while the treatment effects in the three Negative Binomial Models are statistically significant at $5\%$ level of significance, the treatment effect in the \hyperref[PSM]{PSM technique} is not.

A sensitivity analysis based on the basic Negative Binomial Model shows that we observe a significant treatment effect only after accounting for unobserved heterogeneity, underscoring the need to control for the potential systematic variations in the unobservable component of individuals' health production functions.  Because our econometric models of interest are multivariable, the empirical results also show that males, rural residents, individuals with learning disabilities and those with comorbidities have, on average, a longer length of stay in hospital than their counterparts.

The present work builds on the studies in the health economics literature that investigate patients' length of stay in hospital (see, for example, \cite{r20,r40,r21,r45,r22,r41}) and could be extended by modeling several variables in addition to telecare that are endogenously determined with the length of stay in hospital.  Future work could also test the theoretical model and estimation strategy developed in this paper on a non--parametric estimator that does not assume any particular distribution for the data generating process such as the one employed by \cite{r41} as well as model other outcome measures besides the length of stay in hospital.
\section*{Acknowledgements}
I am grateful to the Economic and Social Research Council (Great Britain) for funding my Ph.D. studentship and to the Health Economics Research Unit at the University of Aberdeen for their support during my studies.  I would also like to acknowledge NHS National Services Scotland for providing the linked administrative health and social care data for this study, and eDRIS for providing technical support while conducting the data analysis in the National Safe Haven. 

\appendix
\begin{normalsize}
\section{Appendix A: Proofs for main results}\label{Appendix A}
\subsection{Proof of Lemma 1}\label{lemma 1 proof}
\begin{proof}
Because we can write $T$ as $T=(.)1/p_{T}+C_{1}+\mu$ and $T=(.)p_{\mathcal{B}Q}+C_{2}+\mu$ such that (i) $C_{1},C_{2}\in\mathbb{R}$; (ii) $\mu,p_{T}\not\in C_{1}$; (iii) $\mu,p_{\mathcal{B}Q}\not\in C_{2}$; (iv) $C_{1},C_{2}=[-\infty,\infty]$, and (v) $\mu\in T,H=H(.)$ is unobservable; there exists a function $f_{1}\in\mathcal{F}:p_{T}\mapsto T$ for all $T\in\mathcal{S}$ and $p_{T}\in p$, and also $\exists f_{2}\in\mathcal{F}:p_{\mathcal{B}Q}\mapsto T\forall T\in\mathcal{S}$ and $p_{\mathcal{B}Q}\in p$ so that $T$ may be expressed as $T=t(p_{T})$ and $T=t(p_{\mathcal{B}Q})$.  Similarly, since we also have $T=(.)w+C_{3}+\mu$ such that $\mu,w\not\in C_{3}$ and $C_{3}\in\mathbb{R}=[-\infty,\infty]$; there exists $f_{3}\in\mathcal{F}:w\mapsto T$ for all $T\in\mathcal{S}$ and $w\in\mathcal{W}$ so that $T$ is represented as $T=t(w)$.  And because we also have $T=(.)\mu+C_{4}$ such that $C_{4}\in\mathbb{R}$ and $\mu\not\in C_{4}$, $\exists f_{4}\in\mathcal{F}:\mu\mapsto T\forall T\in\mathcal{S}$ and $\mu\in H=H(.)$ so that $T$ may also be written as $T=t(\mu)$.  Accordingly, we have that $T=t(p_{T},p_{\mathcal{B}Q},w,\mu)$.  To see this for $w,p_{T},p_{\mathcal{B}Q}\in\{-\infty,\infty\}$, let $\lim\limits_{w,p\to\infty}T=\tau_{1}$ and $\lim\limits_{w,p\to -\infty}T=\tau_{2}$.  We expect that $p=\{p_{T},p_{\mathcal{B}Q}\},w,\mu\in\tau_{1}$ and $\tau_{2}$ is such that $T\propto 1/p_{T}$, $T\propto p_{N}$, $T\propto p_{Q}$, $T\propto p_{P}$, $T\propto w$ and $T\propto\mu$.  By Assumption \ref{Assumption 2}, we know that $\forall S^{'}(p,w)\in\mathcal{S}$, if $T\succcurlyeq Q$, then $T$ is substituted for $Q$.  Since we have that $T=(.)1/p_{T}+C_{1}+\mu$ and $T=(.)p_{\mathcal{B}Q}+C_{2}+\mu$, it follows then that $f_{1}:-\Delta p_{T}\Rightarrow +\Delta T$ and $f_{2}:+\Delta p_{\mathcal{B}Q}\Rightarrow+\Delta T$.  But because we also have that $w=\overline{w}$ in Assumption \ref{Assumption 3} for both users of $T$ and $Q$, $p_{T}=p_{Q}$ all else equal.  Consequently, we have that $-\Delta p_{T}\Rightarrow+\Delta w\Rightarrow+\Delta T$ if $T\succcurlyeq Q$ and $-\Delta p_{Q}\Rightarrow+\Delta w\Rightarrow+\Delta Q$ if $Q\succcurlyeq T$.  $T$ and $Q$ can, therefore, be said to be normal since $p_{T}=p_{Q}\Rightarrow T,Q\in\mathcal{K}\propto 1/(p_{T},p_{Q})$ and, as such, $+\Delta w\Rightarrow+\Delta T,Q\in\mathcal{K}$.
\end{proof}
\subsection{Proof of Corollary 1}\label{Corollary 1 proof}
\begin{proof}
We use an analytical approach to provide the proof.  Consider a family of datasets $G\in\mathbb{R}_{m\times n}$ such that $\forall G^{'}\in\mathcal{G}\exists x\prime s\in G^{'}$ so that $G_{i}\in G=X_{m\times i}$ with $\operatorname{rank}(G_{i})=i$ for $i=1,\ldots,n$ and $i\in\mathbb{Z}^{+}$.  By Assumption \ref{Assumption 1}, let $X\in G^{'}$ represent $P,N,T,Q\in\mathcal{S}$.  Consider also a hypothetical outcome variable $Y\in\mathbb{N}^{*}\cup\{0\}$ such that $Y=\alpha+\mathrm{b}X+\epsilon$ is the empirical specification of $V=v(.)$ and $G_{n}=X_{m\times n}$ explains $Y$ so well that $Y\sim\hat{Y}$.  Let also $X\in G^{'}, Y\sim\mathcal{N}(0,1)$ and $Y\perp X\in G^{'}$ for expositional simplicity.  We first use $X_{m\times 1}$ with $\rho(X_{m\times 1})=1$ to estimate $Y$ such that $Y=\alpha_{1}+X_{m\times1}.\mathrm{b}_{1}+\epsilon_{{\hat{Y}_{1}}}$ and $\mathrm{b}_{1}=(X_{m\times1}^{\mathrm{T}}X_{m\times1})^{-1}X_{m\times1}^{\mathrm{T}}Y$.  Since $G_{1}\subset G_{n}$, we expect that $Y-\hat{Y}_{1}=\epsilon_{{\hat{Y}_{1}}}\neq0$.  We then pick $X_{m\times2}=G_{2}$ with $\operatorname{rk}(X_{m\times2})=2$ and estimate $Y$ such that $Y=\alpha_{2}+X_{m\times2}.\mathrm{b}_{2}+\epsilon_{{\hat{Y}_{2}}}$ and $\mathrm{b}_{2}=(X_{m\times2}^{\mathrm{T}}X_{m\times2})^{-1}X_{m\times2}^{\mathrm{T}}Y$.  Given that $\vert X_{m\times2}\vert\textgreater\vert X_{m\times1}\vert$ and $G_{1},G_{2}\subset G_{n}$, we have that $Y-\hat{Y}_{2}=\epsilon_{{\hat{Y}_{2}}}<\epsilon_{{\hat{Y}_{1}}}$.  We then pick $X_{m\times3}$ with full rank and estimate $Y$ such that $Y=\alpha_{3}+X_{m\times3}.\mathrm{b}_{3}+\epsilon_{{\hat{Y}_{3}}}$ and $\mathrm{b}_{3}=(X_{m\times3}^{\mathrm{T}}X_{m\times3})^{-1}X_{m\times3}^{\mathrm{T}}Y$.  Because $\vert X_{m\times3}\vert\textgreater\vert X_{m\times2}\vert\textgreater\vert X_{m\times1} \vert$ and $G_{1},G_{2},G_{3}\subset G_{n}$, it follows then that $Y-\hat{Y}_{3}=\epsilon_{{\hat{Y}_{3}}}<\epsilon_{{\hat{Y}_{2}}}<\epsilon_{{\hat{Y}_{1}}}$.  We continue picking $G^{'}\in\mathcal{G}$ such that $\vert X_{m\times k-1}\vert<\vert X_{m\times k}\vert$ and computing $\epsilon_{{\hat{Y}}}$ for each $G_{i}\in G$ up to $G_{n-1}$.  Since $\vert X_{m\times n-1}\vert<\vert X_{m\times n}\vert$, we have that $\epsilon_{{\hat{Y}_{n-1}}}$ is written as

\begin{equation*}
\begin{pmatrix}
\epsilon_{1} \\
\epsilon_{2}\\
\vdots \\
\epsilon_{m}
\end{pmatrix}=
\begin{pmatrix}
Y_{1} \\
Y_{2}\\
\vdots \\
Y_{m}
\end{pmatrix}-
\begin{pmatrix}
1 & X_{11} & X_{12} & \cdots & X_{1n-1}\\
1 & X_{21} & X_{22} & \cdots & X_{2n-1}\\
\vdots & \vdots & \vdots & \ddots & \vdots\\
1 & X_{m1} & X_{m2} & \cdots & X_{mn-1}
\end{pmatrix}
\begin{pmatrix}
\mathrm{b}_{0} \\
\mathrm{b}_{1} \\
\vdots \\
\mathrm{b}_{n-1}
\end{pmatrix}
\end{equation*}
and by induction $\epsilon_{{\hat{Y}_{k-1}}}\gg \epsilon_{{\hat{Y}_{k}}}$.  Since we also have that $\operatorname{corr}(T,\mu)\neq0$ in Theorem \ref{Theorem 1}, we know that if we were to estimate $\epsilon_{{\hat{Y}_{n}}}$ using $G_{n}=X_{m\times n}$, then $Y-\hat{Y}_{n}$ must contain $\mu$ and we would thus not have $Y\sim\hat{Y}_{n}$.  Accordingly, $\Vert E(Y\vert X)-E(\hat{Y}\vert X) \|$ must also have $\mu$ in addition to $\epsilon\forall G_{i}\in G$ implying that $\epsilon\perp\mu$.  We can, therefore, write $\Vert E(Y\vert X)-E(\hat{Y}\vert X) \|=\delta=\{\epsilon,\mu\}$.
\end{proof}
\subsection{Proof of Proposition 1}\label{Proposition 1 proof}
\begin{proof}
We provide the proof by contradiction.  Conversely, let $Y\perp\!\!\!\perp T\vert X\forall i$ and $t$, and $E(Y_{T=0}\vert T=1)-E(Y_{T=0}\vert T=0)=0$.  Let also $T_{T=1},T_{T=0}\subset T$ be such that $T_{T=1}\Rightarrow T=1$ and $T_{T=0}\Rightarrow T=0$.  Suppose that $\exists G_{T=1}\in\mathbb{R}_{m\times n}$ such that $\hat{Y}_{T=1}$ estimated by $G_{T=1}=\{X_{m\times n}\}_{T=1}$ with $\operatorname{rank}(n)$ is the best predictor of $Y_{T=1}$.  Correspondingly, suppose also that $\exists G_{T=0}\in\mathbb{R}_{m\times n}$ such that dataset $G_{T=0}=\{X_{m\times n}\}_{T=0}$ with $\rho(n)$ estimates $Y_{T=0}$ with a great deal of precision so that $Y_{T=0}\sim\hat{Y}_{T=0}$.  Then by the definition of $V=v(.)$ and $T=t(.)$, it must be the case that $\mu_{T=1}\in T_{T=1}t(.)$ is equal to $\mu_{T=0}\in T_{T=0}t(.)$ if the difference between $E(Y_{T=0}\vert T=1)$ and $E(Y_{T=0}\vert T=0)$ is to be equal to zero.  Since $\mu_{T=1},\mu_{T=0}\in H=H(.)$ is unobservable, however, it is not reasonable to expect that $\mu_{T=1}=\mu_{T=0}$.  Furthermore, by the definition of $V=v(.)\Rightarrow Y=f(X,T)$ where $Y$ is our outcome measure of interest we know that $\operatorname{corr}(Y,\mu)\neq0$ and thus the claim that $Y\perp\!\!\!\perp T\vert X\forall i$ and $t$ cannot be true.  Moreover, by Assumption \ref{Assumption 1} we know that $\exists P,N\in T=t(.)$ since $p_{S}\in S^{'}$ and because by Assumption \ref{Assumption 2} we also have that $T\succcurlyeq Q\Rightarrow T$ is chosen instead of $Q$, it is highly plausible that $\exists m\in T=t(.)$ such that $m\not\in P,N$ and $\operatorname{corr}(m,Y)\neq0$.  Accordingly, we must have that $Y$ is not independent of $T$ given $X$ i.e. $Y\not\!\perp\!\!\!\perp T\vert X\forall i$ and $t$.
\end{proof}
\subsection{Proof of Proposition 2}\label{Proposition 2 proof}
\begin{proof}
Consider a population $\mathrm{N}$ with $i$ individuals $:\mathrm{N}\in\mathbb{N}=\{1,2,\ldots,i\}$.  Since we have that $T\in\mathbb{Z}_{2}:\mathbb{Z}_{2}=\{0,1\}$ and $Y=\{0,1,2,\ldots,n\}\in\mathbb{N}^{0}$ for all $i$ individuals in population $\mathrm{N}$, let $T_{T=1}$ represent $T=\{1\}\forall i\in \mathrm{N}$ and $T_{T=0}$ represent $T=\{0\}\forall i\in \mathrm{N}$.  Correspondingly, let also $Y_{T=1}$ represent $Y:T=\{1\}$ and $Y_{T=0}$ represent $Y:T=\{0\}\forall i\in \mathrm{N}$.  It follows then that $Y_{T=1},Y_{T=0},T_{T=1}$ and $T_{T=0}$ are such that $Y_{T=1},Y_{T=0}\in Y\forall i$ and $t$ in population $\mathrm{N}$ and $T_{T=1},T_{T=0}\subset T\forall i$ and $t$ in population $\mathrm{N}$.  Since $G_{TX}=\{X_{m\times n},T_{m\times 1}\}$ with $\rho(n)$ explains $Y$ almost perfectly, we must have that $\vert Y\vert\sim\vert\hat{Y}\vert$ which implies that $E(Y_{T=0}\vert T=1)\sim E(Y_{T=0}\vert T=0)$ in order for $\mathrm{s}_{1}=\{Y\vert X,T\}$ and $\mathrm{s}_{2}=\{Y\vert X,\hat{T},t\}$ to be sufficient.  But by Theorem \ref{Theorem 1}, we have that $\mu\in V=v(.)$ and $\mu\in T=t(.)$.  And by Corollary \ref{Corollary 1}, we also have that $\epsilon\perp\mu\Rightarrow\mu\in V=v(.)\Rightarrow\mu\in Y$ even though $\epsilon\not\in\hat{T}$.  Accordingly, $Y\not\!\perp\!\!\!\perp T\vert X$ and $Y\not\!\perp\!\!\!\perp\hat{T}\vert X$.  We, therefore, cannot say that $\mathrm{s}_{1}=\{Y\vert X,T\}$ and $\mathrm{s}_{2}=\{Y\vert X,\hat{T},t\}$ are sufficient.  Moreover, by the definition of $\mathrm{s}_{1}$, we must have that $Y_{t-1}\equiv Y_{t}\Rightarrow V_{t-1}\equiv V_{t}$ since it is applicable $\forall i$ and $t$ in population $\mathrm{N}$.  Because by the definition of $T=t(.)$ we have that $\mu\in T=t(.)$ for all $i\in\mathrm{N}$, it is reasonable to expect that $\exists m\in\mathrm{N}$ such that $T_{m\in\mathrm{N}}=t(.)\neq T_{m\not\in\mathrm{N}}=t(.)$ and consequently that $[Y_{m\in\mathrm{N}}=f(X,T,t)]\neq [Y_{m\not\in\mathrm{N}}=f(X,T,t)]$.  Furthermore, from $I=g(X,Z,t)$ we have that $Y=f(X,T,t)\iff y\in Y:f(y)\neq0$.  Accordingly, $\mathrm{s}_{3}$ can be sufficient if and only if $E(T,\xi)\approx 0$ and $E(Y\vert I=1)\sim E(Y\vert I=0)\Rightarrow E(I,\xi_{s})\approx 0$, all else equal.  Strategy $\mathrm{s}_{4}$ can be sufficient iff $E(Y\vert I=1)-E(Y\vert I=0)\approx 0\Rightarrow E(I,\xi_{s})\approx 0$, ceteris paribus.
\end{proof}
\end{normalsize}

\end{document}